\documentclass[12pt]{article}
\usepackage{hyperref}
\usepackage{float}
\usepackage{amsmath}
\usepackage{graphicx,psfrag,epsf}
\usepackage{enumerate}
\usepackage{natbib}
\usepackage{multirow}
\usepackage[ruled,vlined]{algorithm2e}
\usepackage{array}
\usepackage{graphicx}
\usepackage[utf8 ]{inputenc}
\usepackage[english]{babel}
\usepackage{amsthm}
\usepackage[bottom]{footmisc}
\usepackage{longtable}
\usepackage{mathrsfs}
\usepackage{extarrows}
\usepackage{bigints}
\usepackage{subcaption}

\usepackage{setspace}

\usepackage[top=0.78in, bottom=0.78in, left=0.78in, right=0.78in]{geometry}
\usepackage{amsmath}
\usepackage{amsfonts}

\newtheorem{theorem}{Theorem}[section]

\newtheorem{assumption}{Assumption}

\DeclareTextFontCommand{\textmyfont}{\footnotesize}

\newenvironment{example}[1][Example]{\begin{trivlist}
\item[\hskip \labelsep {\bfseries #1}]}{\end{trivlist}}
\newenvironment{remark}[1][Remark]{\begin{trivlist}
\item[\hskip \labelsep {\bfseries #1}]}{\end{trivlist}}

\newcommand\floor[1]{\lfloor#1\rfloor}

\begin{document}

\title{Modeling Spectral Properties in Stationary Processes of Varying Dimensions with Applications to Brain Local Field Potential Signals 
\footnote{AMS subject classification. Primary: 62M10. Secondary: 62M15.}
\footnote{Keywords and phrases: Multivariate time series, nonstationary, spectral matrix, local field potential}
\footnote{This work is support in part by KAUST, NIH NS066001, Leducq Foundation 15CVD02
and NIH MH115697. }}

\author{
Raanju Ragavendar\ Sundararajan \\ Southern Methodist University  \\
Ron D. Frostig \\ University of California, Irvine  \\
\and Hernando Ombao\\ King Abdullah University of Science and Technology}


\date{}

\maketitle

\newpage

\begin{abstract}

\noindent A common class of methods for analyzing of multivariate time series, stationary and 
nonstationary, decomposes the observed series into latent sources. Methods such as principal 
compoment analysis (PCA),  independent component analysis (ICA) and Stationary Subspace Analysis (SSA) 
 assume the observed multivariate process is generated by latent sources that are stationary or 
nonstationary. We develop a method that tracks changes in the complexity of a 32-channel 
local field potential (LFP) signal from a rat following an experimentally induced stroke. We study 
complexity through the latent sources and their dimensions that can change  across epochs 
due to an induced shock to the cortical system. 
Our method compares the spread of spectral information in several multivariate stationary processes with different dimensions. A frequency specific spectral ratio (FS-ratio) statistic is proposed and its asymptotic properties are derived. The FS-ratio is 
blind to the dimension of the stationary process and captures the proportion of spectral information in various (user-specified) frequency bands. We 
apply our method to study differences in complexity and structure of the LFP before and after system 
shock. The analysis indicates that spectral information in the 
beta frequency band (12-30 Hertz) demonstrated the greatest change in structure and 
complexity due to the stroke.
  
\end{abstract}

\hrule
\hrulefill

\section{Introduction}

A common class of methods for modeling multivariate time series data decomposes the observed series into latent sources that can be stationary or nonstationary. The goal in this paper is to develop a method that tracks 
changes in the complexity of signals following a shock that is induced on a biological system. In particular,  
the proposed method will be used to study changes in the rat's brain functional network resulting from an 
induced stroke in an experiment conducted by co-author (R. D. Frostig) at the Neurobiology laboratory at 
UC Irvine. Here we shall characterize complexity in local field potentials (LFPs) through the latent sources 
and their evolving dimension. Figure \ref{fig:ratbrain} below depicts the rat's cortex and the locations of the 
32 sensors implanted on the cortical surface from which the LFP signal is recorded. This 32-dimensional signal 
is our observed time series.  

\begin{figure}[H] 
\centering
\includegraphics[scale=0.95]{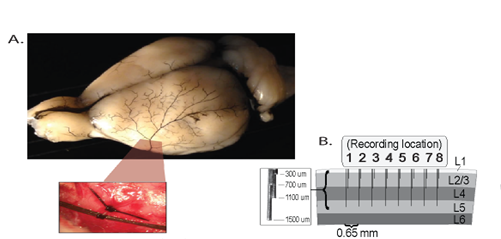}
\caption{ Visual representation of the 32 microelectrodes on the rat's cortex from which the local field potential (LFP) signal is recorded. The distance between microelectrodes is 0.65mm and the total distance between microelectrode 1 and microelecteode 8 is 3.9mm.} \label{fig:ratbrain}
\end{figure}

\noindent  The local field potential signals from the experiment will be modeled as
\begin{equation}\label{e:decomp_neuro_exp}
X_{i,t} \; = \; A_i Y_{i,t} \; + \; Z_{i,t},
\end{equation}
where $i$ is the indicator of the epoch ($i=1,2,\cdots,N$), $A_i$ is the unknown mixing matrix for epoch $i$, $Y_{i,t} \in \mathbb{R}^{d_i}$ are the latent sources of interest in epoch $i$ (an epoch is a 1-second block of 
LFP) and $Z_{i,t}$ is the nonstationary sources. The interest in obtaining the latent $Y_{i,t}$ can be viewed from different perspectives depending on the end objective of the statistical problem.  A few examples include the classical dynamic PCA for time series from \citet{brillinger81}, PCA in the multivariate time series setting (\citet{stockwatson}, \citet{slex_ombao}, \citet{ombao_2006}, \citet{yao2018}), factors models and ICA (\citet{lam2012}, \citet{mattesondoc}, \citet{ombao_motta_2012}). The aim of these current methods 
is primarily in simplifying the analysis of multivariate time series $X_{i,t}$ in \eqref{e:decomp_neuro_exp} 
by producing summaries which are a few useful independent/orthogonal components or factors $Y_{i,t}$. 
Stationary subspace analysis (SSA), introduced by \citet{ssa09} and studied further by \citet{sundararajan:2017}, 
is another related method that decomposes an observed multivariate nonstationary time series $X_{i,t}$ 
into stationary $Y_{i,t} \in \mathbb{R}^{d_i}$ and nonstationary $Z_{i,t} \in \mathbb{R}^{p-d_i}$ 
components. However, unlike PCA and ICA, the latent components in SSA are not constrained to be 
independent/orthogonal. This is a major advantage because it gives a more realistic (less constrained) 
description of observed brain processes. Furthermore, the SSA framework would rightly treat the observed 
brain signal as a nonstationary process ( \citet{slex_ombao}, \citet{srinivasan_2003}, 
\citet{srinivasan_2006}, \citet{SSAbci}, \citet{wu_2016},  \citet{gao_2018}, \citet{euan_2019}). 

Irrespective of whether one is interested in PCA, factor modeling, ICA, SSA, the dimension $d_i$ of these 
latent sources $Y_{i,t}$ should be allowed to change across  $i=1, 2, \ldots, N$ epochs. Artificially setting 
the dimension to be the same across the epochs results in loss of useful information since these changes 
could be indicative of useful brain process such as learning (\citet{fiecas_2016}). Indeed 
brain processes evolve across the entire recording period (\citet{fiecas_2016}, \citet{ombao_2018}) and 
thus $d_i$ should be allowed to change across epochs $i$. Moreover, the evolution of the $d_i$ can itself 
serve as a feature in understanding how the brain function evolves during an experiment. 

The application that motivates our methodology is the analysis of local field potentials (LFP) in an experiment that simulates ischemic stroke in humans.\footnote{Data source: Stroke experiment conducted in the lab of co-author (Ron Frostig) at his Neurobiology lab; \texttt{http://frostiglab.bio.uci.edu/Home.html}}. The dataset comprises of 600 epochs worth of LFP recordings (each epoch is 1 second long)  from 32 microelectrodes implanted in a rat's cortex. A stroke is induced midway through the experiment (epoch 300) by severing the medial cerebral artery. In Figure \ref{fig:pvalue_component_plot}, we present the p-values from a test of second-order stationarity carried out on each of the $p=32$ microelectrodes at each epoch. We notice that these individual microelectrodes are more stationary after the stroke than before and this shift suggests a varying dimension $d_i$ of $Y_{i,t}$ in model \eqref{e:decomp_neuro_exp}. In Figure \ref{fig:lfp_dimension_plot}, we apply SSA and plot the estimates of the stationary subspace dimension $d_i$ across $N=600$ epochs using the method in \citet{sundararajan2019}. We notice the varying dimension estimates across the 600 epochs thereby making comparison of $Y_{i,t} \in \mathbb{R}^{d_i}$ across $i=1,2,\cdots,N$ epochs difficult. For example, it is non-trivial and challenging to compare the spectrum of $Y_{i,t} \in \mathbb{R}^{d_i}$ and $Y_{j,t} \in \mathbb{R}^{d_j}$ for two different epochs $i$ and $j$ when $d_i \neq d_j$.

\begin{figure}[h]
\centering
\includegraphics[scale=0.42]{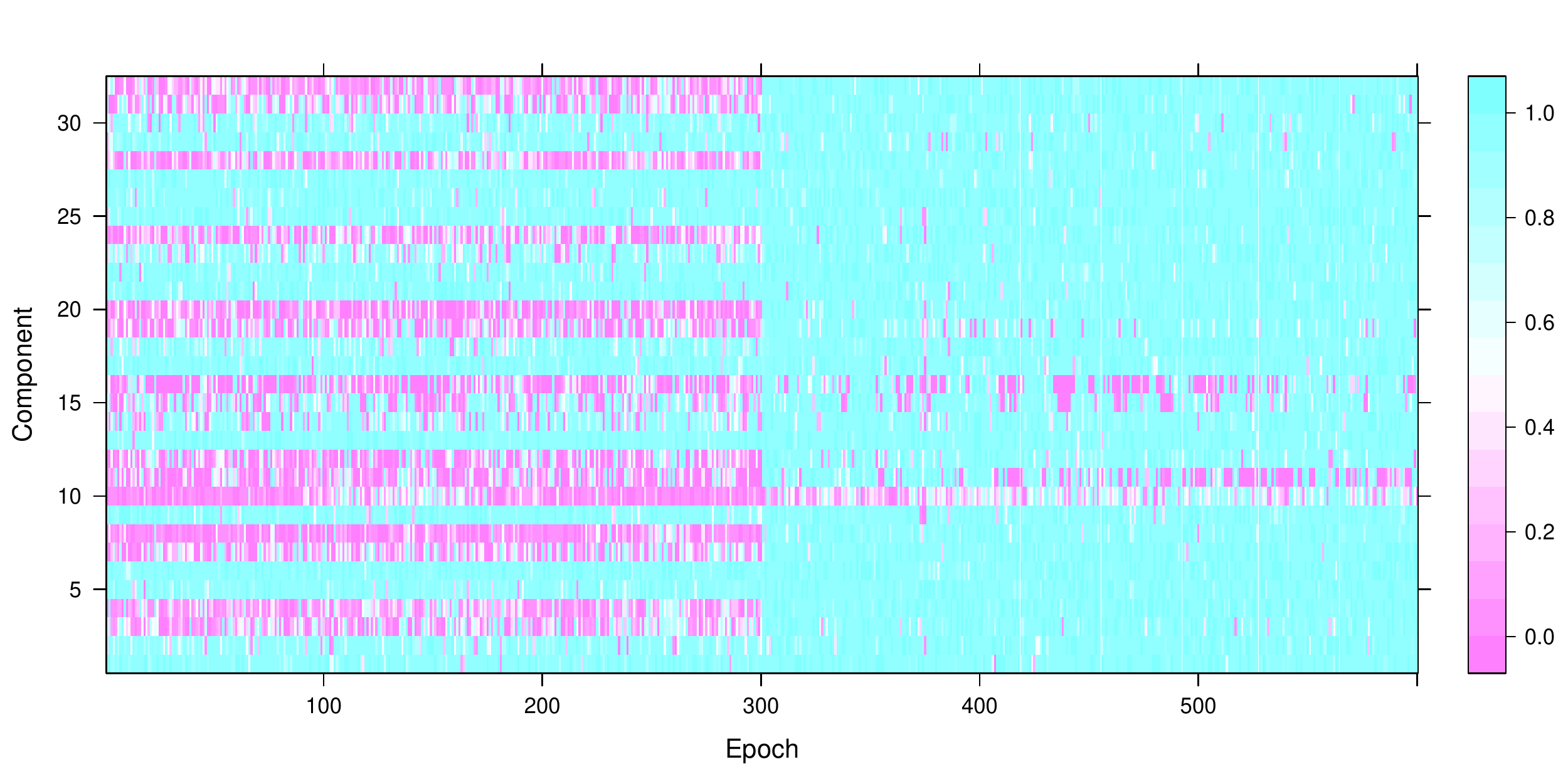}
\caption{ p-values from the test of second-order stationarity on each of the $p=32$ LFP microelectrodes  (y-axis) for all 600 epochs (x-axis).   } \label{fig:pvalue_component_plot}
\end{figure}
\begin{figure}[h]
\centering
\includegraphics[scale=0.42]{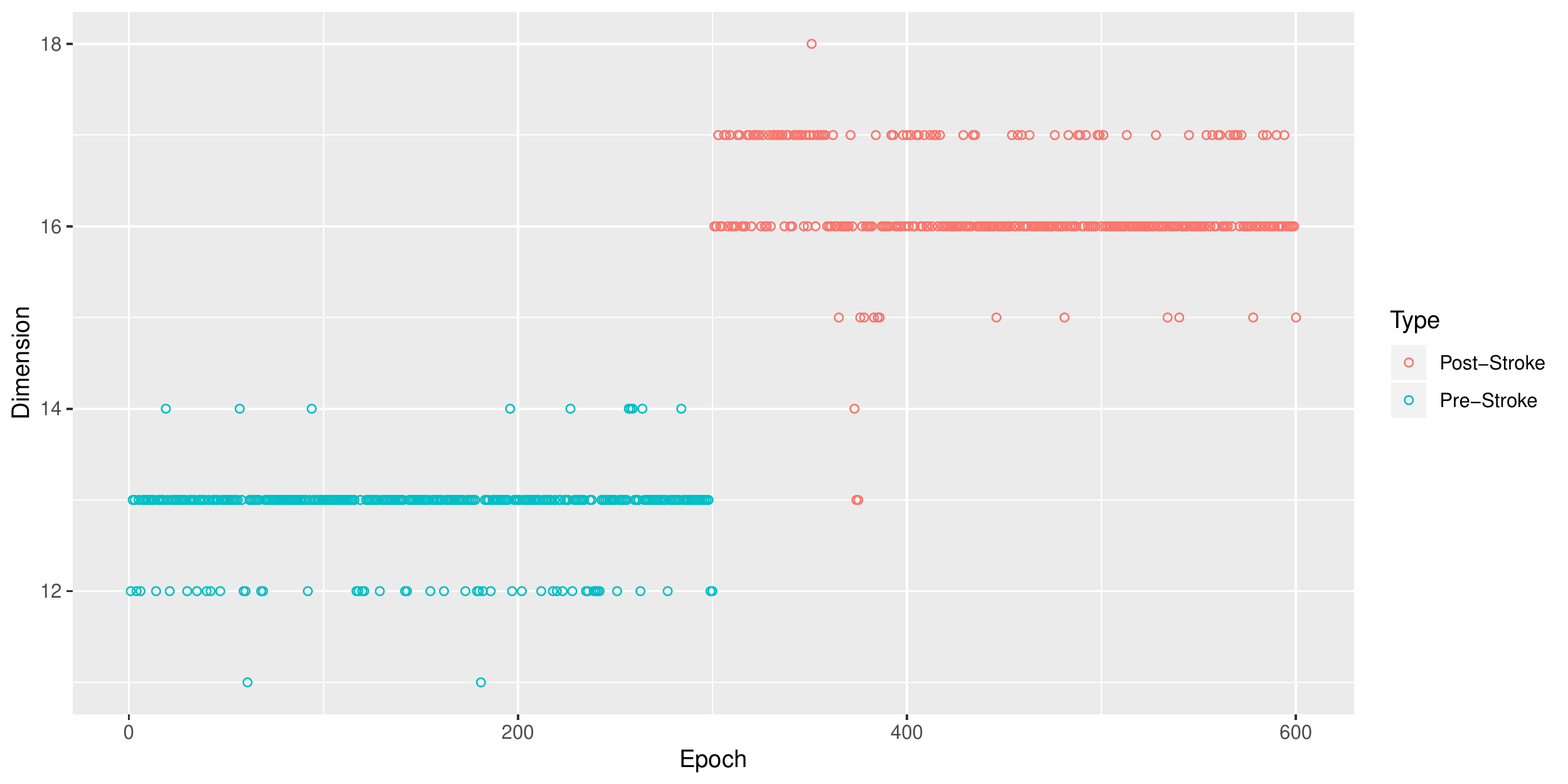}
\caption{ Plot of estimated stationary subspace dimensions $\widehat{d}_i$ for the $i=1,2,\hdots,N=600$ epochs in the stroke experiment.  } \label{fig:lfp_dimension_plot}
\end{figure}

Another related application in neuroscience is functional connectivity wherein the aim is to model dependence between different brain regions at various epochs in an experiment; \citet{cribben_2012}, \citet{cribben_2013}, \citet{cribben_2016}, \citet{cribben_2018}. To mitigate the problem of high-dimensionality arising due to signal from densely voxelated cortical surface, parcellation leads to disjoint regions of interest (ROI) of the brain and signal summaries are obtained in each of these regions. Dependence measures between these ROIs are then computed using their respective signal summaries. In the above pursuit of region-wise comparison of the brain, it is natural to encounter the problem of comparing multivariate processes, say from two different regions, that have unequal dimensions. In \citet{ombao_ting_2016} the problem of modeling effective connectivity in high-dimensional cortical surface signal is pursued wherein a factor analysis is carried out on each ROI and VAR models are used to jointly model the latent factors. Here again, one can potentially end up with unequal number of optimal latent factors 
from different ROIs and thus it will be challenging to make comparison across ROIs. 

Motivated by such applications, we propose a new method to compare  \underline{spectral information} in different multivariate stationary processes of \underline{varying dimensions}. More specifically, the aim is to capture the amount of spectral information in various frequency bands in different stationary processes of unequal  dimensions. 
There are already many methods and models that discuss evolution of spectral information but the key contribution of this paper is in modeling evolution of the spectrum while allowing dimension to also evolve over time. 
We introduce a frequency-specific spectral ratio, which we call the FS-ratio, statistic that accounts for 
the proportion of spectral information in various frequency bands. FS-ratio can be used to (i). identify frequency bands where the pre- and post-stroke epochs are significantly different, (ii). identify frequency bands that accounts for most variation within pre (and post) stroke epochs and (iii). identify the frequency bands that are consistent (vs inconsistent) across all the 600 epochs. One of the key features of this statistic is that it is blind to the 
dimension of the multivariate stationary process and can be used to compare successive epochs with 
possibly different dimensions in the stationary sources. Thus, the proposed FS-ratio is very useful in (a). discriminating between the pre and post stroke signals and (b). tracking changes over the entire course of the experiment while allowing for \underline{varying dimensions}. In Section \ref{s:methodology} we develop our FS-ratio statistic and derive its asymptotic properties. We evaluate the performance of the proposed FS-ratio statistic through some simulation examples in Section \ref{s:simulation}. We return to the LFP dataset in Section \ref{s:application} and discuss the usefulness of the proposed ratio statistic in discriminating between pre- and post-stroke epochs. Section \ref{s:conclusion} concludes.     

The application of our method to the LFP data in Section \ref{s:application} justifiably models the observed 
LFP signals as a multivariate nonstationary time series that is generated by latent sources of interest. 
First, our method clearly demonstrates the evolution of the dimension of these latent sources across the 
600 epochs. Thus, our method provides useful insights on the evolution of the LFP signal. Second, the FS-ratio statistic, having the ability to compare two multivariate processes of unequal dimensions, is estimated and indicates that the beta frequency band information exhibits most variation 
over the course of the stroke experiment.

\section{Methodology}
\label{s:methodology}

In this section we first describe our FS-ratio statistic and the method to analyze the evolution of spectral information in stationary processes with varying dimensions. The asymptotic properties of the proposed statistic along with the required assumptions is discussed in Section \ref{s:theory}.

\subsection{The FS-ratio statistic}
\label{s:spectrum_changes}

Let $Y_{i,t}$, $1 \leq i \leq N$, be a $d_i$-variate  zero-mean second order stationary time series and let $Y_t = (Y_{1,t} , Y_{2,t} , \hdots , Y_{N,T})^{\top}$, $1\leq t \leq T$,  be a $d$-variate zero-mean second-order stationary time series where $d = \sum_{i=1}^N  d_i$. The $d \times d$ spectral matrix for $Y_t$ can be written as 

\begin{equation} \label{e:block-spectrum} 
f(\omega) =  \begin{bmatrix}
f_{11}(\omega) & f_{12}(\omega) & \hdots & f_{1d}(\omega) \\
f_{21}(\omega) & f_{22}(\omega) & \hdots & f_{2d}(\omega) \\
\vdots & &  & \vdots \\
f_{d1}(\omega) & f_{d2}(\omega) & \hdots & f_{dd}(\omega) \\
\end{bmatrix} = \begin{bmatrix}
g_{11}(\omega) & g_{12}(\omega) & \hdots & g_{1N}(\omega) \\
g_{21}(\omega) & g_{22}(\omega) & \hdots & g_{2N}(\omega) \\
\vdots & &  & \vdots \\
g_{N1}(\omega) & g_{N2}(\omega) & \hdots & g_{NN}(\omega) \\
\end{bmatrix} 
\end{equation}
where $\omega \in \lbrack -\pi , \pi \rbrack  $ and $g_{ij}(\omega), \; 1 \leq i,j \leq N $, are $q_{ij} \times q_{ij}$ block matrices matrices with $q_{ij} = \min(d_i,d_j)$. Note that $g_{ij}(\omega) = g_{ji}(\omega)^{*}$, the conjugate transpose. 

The discrete Fourier transform and the periodogram of $Y_t$ are defined in the usual manner,
\begin{equation*}
 J_T(\omega)= \frac{1}{\sqrt{2\pi T}} \sum_{t=1}^{T} Y_t \textrm{exp}(-it\omega),\quad I_{T}(\omega)=J_T (\omega)J_T (\omega)^*,
\end{equation*}
\noindent where $J_T(\omega)^*$ denotes the conjugate transpose. Similar to the representation in \eqref{e:block-spectrum}, the $d \times d$ periodogram matrix $I_{T}(\omega)$ can be veiwed as 
\begin{equation} \label{e:block-periodogram} 
I_T(\omega) =  \begin{bmatrix}
\mathcal{I}_{11}(\omega) & \mathcal{I}_{12}(\omega) & \hdots & \mathcal{I}_{1d}(\omega) \\
\mathcal{I}_{21}(\omega) & \mathcal{I}_{22}(\omega) & \hdots & \mathcal{I}_{2d}(\omega) \\
\vdots & &  & \vdots \\
\mathcal{I}_{d1}(\omega) & \mathcal{I}_{d2}(\omega) & \hdots & \mathcal{I}_{dd}(\omega) \\
\end{bmatrix} = \begin{bmatrix}
I_{11}(\omega) & I_{12}(\omega) & \hdots & I_{1N}(\omega) \\
I_{21}(\omega) & I_{22}(\omega) & \hdots & I_{2N}(\omega) \\
\vdots & &  & \vdots \\
I_{N1}(\omega) & I_{N2}(\omega) & \hdots & I_{NN}(\omega) \\
\end{bmatrix} 
\end{equation}
where $\omega \in \lbrack -\pi , \pi \rbrack  $ and $I_{ij}(\omega), \; 1 \leq i,j \leq N $, are $q_{ij} \times q_{ij}$ block matrices matrices with $q_{ij} = \min(d_i,d_j)$ and $I_{ij}(\omega) = I_{ji}(\omega)^{*}$, the conjugate transpose. 

The estimated $ d \times d $ spectral matrix, for $\omega \in \lbrack -\pi,\pi \rbrack$ is given by
\begin{equation} \label{eq:estimated_spectral_matrix}
\hat{f}(\omega) = \frac{1}{T} \sum_{j=-\floor{\frac{T}{2}}+1}^{\floor{\frac{T}{2}}} \;  \; K_h( \omega - \omega_j) \;  I_{T}(\omega_j),
\end{equation}
where  $\omega_j=\frac{2 \pi}{T} j$ and $K_h(\cdot)=\frac{1}{h}K(\frac{\cdot}{h})$ where $K(\cdot)$ is a nonnegative symmetric kernel function and $h$ denotes the bandwidth. Assumptions on the kernel and bandwidth
to ensure uniform consistency in $\omega \in \lbrack -\pi, \pi \rbrack $
 of the estimated spectral matrices are listed in Section \ref{s:theory}. 

The aim of this work is to compare the $d_i \times d_i $ spectral matrices $g_{ii}(\omega)$ across $i=1,2,...N$ epochs over a specific frequency range $(a,b)$ for some $0<a<b<\pi$. The challenge here, however, is that the dimension of the processes $Y_{i,t} \in \mathbb{R}^{d_i}$  varies across the $N$ epochs and hence the spectral matrices across $N$ epochs have varying dimensions. We thus focus on the spread or distribution of spectral information in each of these stationary processes $Y_{i,t}$ across different frequency ranges. More precisely, we define the frequency-specific spectral (FS-ratio) parameter  as   

\begin{equation} \label{e:R-def-pop}
R_{i,a,b} \; = \; \frac{2 \;  r_{i,a,b}}{r_{i,-\pi,\pi}} \;  =  \; \frac{ 2 \int_a^b || vec( g_{ii} (\omega) ) ||_2^2 d \omega}{\int_{-\pi}^{ \pi} ||  vec( g_{ii}(\omega) ) ||_2^2 d \omega} 
\end{equation}
for some frequency band $(a,b) \subset (0, \pi)$,  for $i=1,2,...N$. Observe that $R_{i,a,b} \in (0,1)$ can be viewed as a measure that captures the proportion of spectral information found in the frequency range $(a,b)$. 

The data analogue of the FS-ratio parameter in \eqref{e:R-def-pop} is then given by the FS-ratio statistic:

\begin{equation} \label{e:R-def-samp}
\widehat{R}_{i,a,b} \; = \; \frac{2 \; \hat{r}_{i,a,b}}{\hat{r}_{i,-\pi,\pi}} \; = \;  \frac{ 2 \int_a^b || vec( \widehat{g}_{ii}(\omega) ) ||_2^2 d \omega}{\int_{-\pi}^{ \pi} ||  vec( \widehat{g}_{ii} (\omega) ) ||_2^2 d \omega} 
\end{equation}
 for some $0 < a < b < \pi $ for $i=1,2,...N$. The asymptotic properties of the quantities $\hat{r}_{i,a,b}$ and $\widehat{R}_{i,a,b}$ are discussed in Section \ref{s:theory}. Before proceeding further, we provide a simple illustration of FS-ratio statistic through the following example.

\begin{example}[Example 2.1]
We consider univariate process $Y_{i,t}$ that is given by 
\begin{equation}\label{e:example_2_1_model}
Y_{i,t} = 1_{i<300} \Big( 0.9 Y_{i,t-1}  + \epsilon_{i,t} \Big) \;+ \; 1_{i \geq 300} \Big( 0.25 Y_{i,t-1} -0.75 Y_{i,t-2} + \epsilon_{i,t} \Big)
\end{equation}
where $\epsilon_{i,t}$ is i.i.d $N(0,1)$, $i=1,2,\hdots,N=600$ epochs, $t=1,2,\hdots,T=1000$. The process $Y_{i,t}$ is given by an AR(1) with coefficient 0.9 or by an AR(2) with coefficients (0.25,-0.75). The top panel in Figure \ref{fig:ar_spec_spread} plots the true AR(1) and AR(2) spectrum from \eqref{e:example_2_1_model} respectively. The bottom panel in Figure \ref{fig:ar_spec_spread} plots the FS-ratio statistic $\widehat{R}_{i,a,b}$ for different frequency ranges $(a,b) \subset (0,\pi)$. When $(a,b)=(0,\pi/10)$, $\widehat{R}_{i,a,b}$ is almost 100\% percent for epochs $i=1,2,\hdots,300$ because the AR(1) spectrum with coefficient 0.9 has a lot of low frequency information. Similarly when $(a,b) = (2 \pi /5, 3 \pi / 5 )$, we get $\widehat{R}_{i,a,b}$ to be around 85\% for epochs $i=301,302,\hdots,600$ as the AR(2) with coefficients $(0.25,-0.75)$ has a lot of spectral information in that frequency range. 

Next, we consider univariate process $Y_{i,t}$ that is given by 
\begin{equation}\label{e:example_2_2_model}
Y_{i,t} = 1_{i<300} \Big( -0.9 Y_{i,t-1}  + \epsilon_{i,t} \Big) \;+ \; 1_{i \geq 300} \Big( 0.25 Y_{i,t-1} -0.75 Y_{i,t-2} + \epsilon_{i,t} \Big)
\end{equation}
where $\epsilon_{i,t}$ is i.i.d $N(0,1)$, $i=1,2,\hdots,N=600$ epochs, $t=1,2,\hdots,T=1000$. The only change here is that for $i<300$, the AR(1) coefficient is -0.9 instead of 0.9 in \eqref{e:example_2_1_model}. For the model in \eqref{e:example_2_2_model}, similar to Figure \ref{fig:ar_spec_spread}, we obtain Figure \ref{fig:ar_spec_spread_1}. Here the AR(1) with coefficient -0.9 has a lot of high frequency information and hence when $(a,b)=( 4\pi / 5 , \pi )$ we see that for epochs $i=1,2,\hdots,300$, the FS-ratio statistic is close to 100\%. 

\end{example}

\begin{figure}[H]
\begin{center}
\includegraphics[scale=0.45]{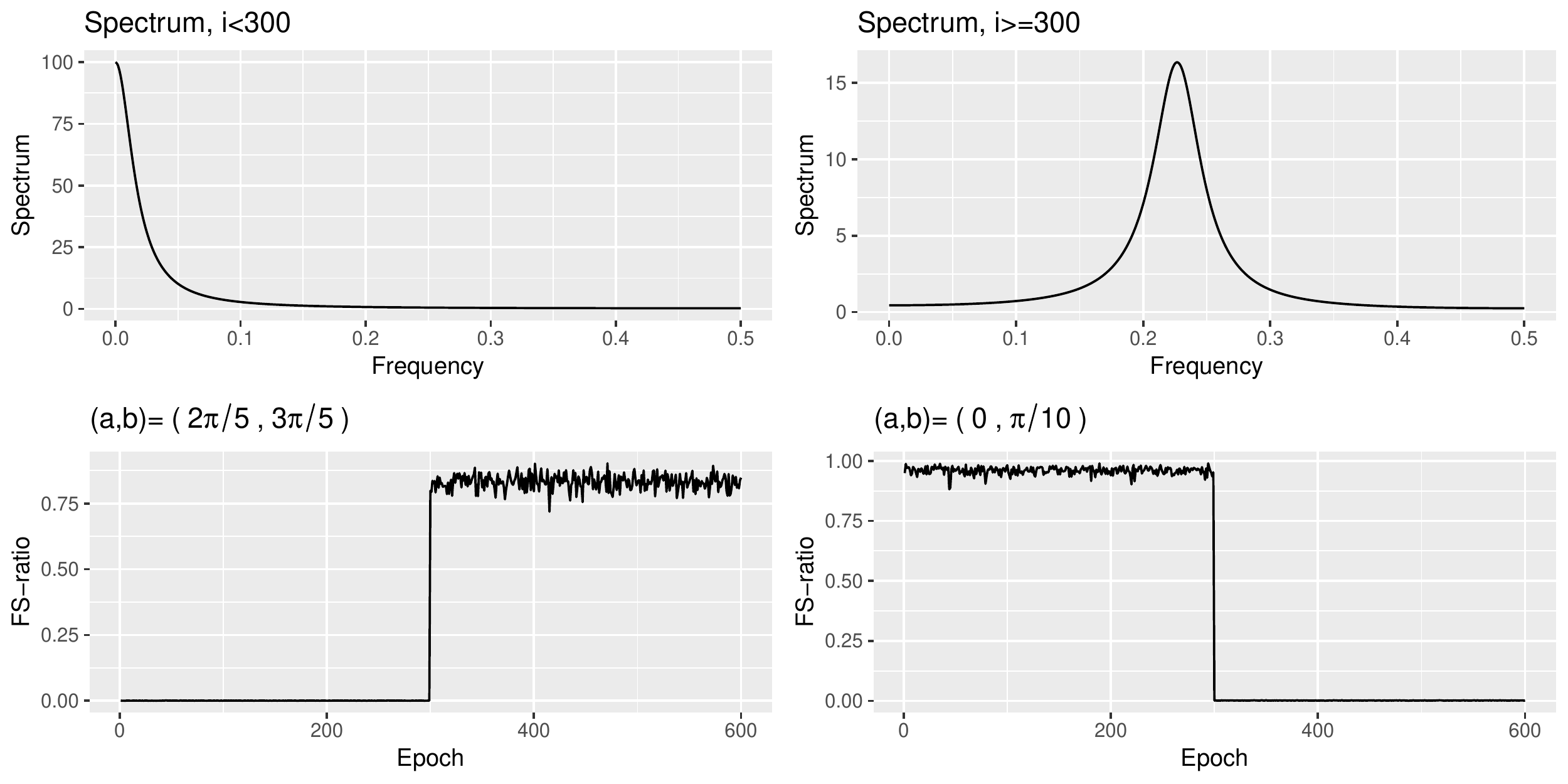}
\caption{  \textbf{Example 2.1} Top panel: Plots of the true AR(1) and AR(2) spectrum from \eqref{e:example_2_1_model} respectively; Bottom panel: Plot of the FS-ratio statistic $\widehat{R}_{i,a,b}$ for $i=1,2,\hdots,N=600$ for specified frequency ranges $(a,b)$.  } \label{fig:ar_spec_spread}
\end{center} 
\end{figure} 

\begin{figure}[H]
\begin{center}
\includegraphics[scale=0.45]{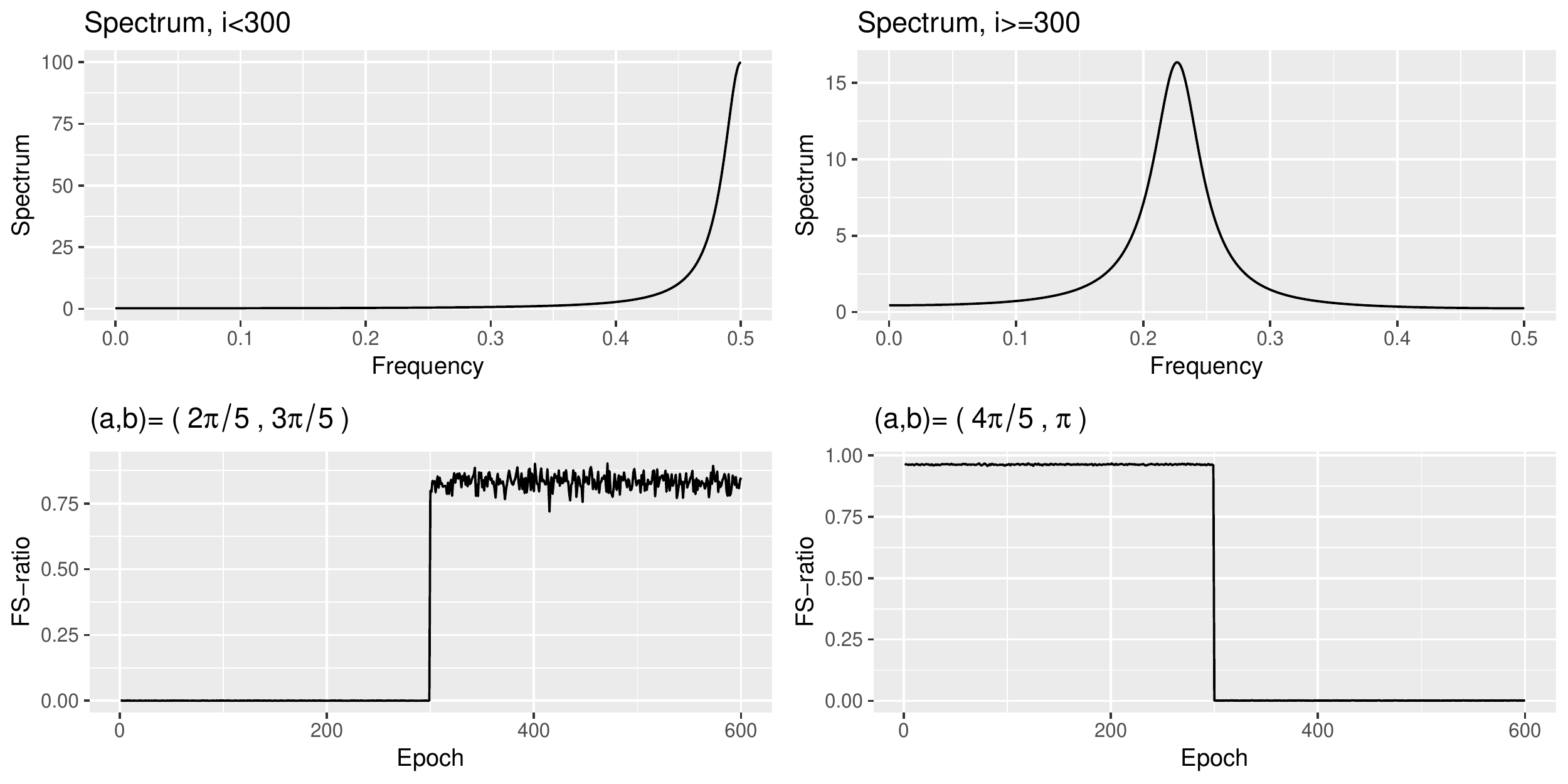}
\caption{ \textbf{Example 2.1} Top panel: Plots of the true AR(1) and AR(2) spectrum from \eqref{e:example_2_2_model} respectively; Bottom panel: Plot of the FS-ratio statistic $\widehat{R}_{i,a,b}$ for $i=1,2,\hdots,N=600$ for specified frequency ranges $(a,b)$.  } \label{fig:ar_spec_spread_1}
\end{center} 
\end{figure} 

\subsubsection{Theoretical properties of the FS-ratio statistic}
\label{s:theory}
In this section we list the required assumptions and discuss the asymptotic properties of the statistics $\hat{r}_{i,a,b}$ and FS-ratio $\widehat{R}_{i,a,b}$. 

\begin{assumption}\label{eq:dgpmix}
Let $\lbrace Y_t \rbrace, \; t \in \mathbb{Z}$ be a $d$-variate zero-mean second-order stationary time series. For any $k>0$, the $k^{th}$ order cumulants of $Y_t$ satisfy
\vspace{-0.2cm}
\begin{equation*}
\sum_{u_1,u_2,...,u_{k-1} \in \mathbb{Z}} \lbrack \; 1+ |u_j|^2 \; \rbrack \;| c_{a_1,a_2,...,a_k}(u_1,u_2,...,u_{k-1})| \; < \infty
\vspace{-0.2cm}
\end{equation*}
for $j=1,2,...,k-1$ and $a_1,a_2,...,a_k = 1,2,...,d$ where $c_{a_1,a_2,...,a_k}(u_1,u_2,...,u_{k-1})$ is the $k^{th}$ order joint cumulant of $Y_{a_1,u_1},...,Y_{a_{k-1},u_{k-1}},Y_{a_k,0}$ as defined in \cite{brillinger81}.
\end{assumption}

\begin{assumption}
(a). The kernel function $K(\cdot)$ is bounded, symmetric, nonnegative and Lipschitz-continuous with compact support $\lbrack -\pi ,  \pi \rbrack$ and
\vspace{-0.2cm}
\begin{equation*}
\int_{-\pi}^{ \pi} K(\omega) \textrm{d}\omega = 1.
\vspace{-0.2cm}
\end{equation*}
where $K(\omega)$ has a continuous Fourier transform  $k(u)$ such that
\vspace{-0.2cm}
\begin{equation*}
\int k^2(u)\textrm{d}u < \infty \;\; \textrm{and} \;\; \int k^4(u)\textrm{d}u < \infty.
\vspace{-0.2cm}
\end{equation*}
(b). The bandwidth $h$ is such that $h^{9/2}T \rightarrow 0$ and $h^2 T \rightarrow \infty$ as $T \rightarrow \infty$.
\end{assumption}

\begin{remark}[Remark 2.2]
\begin{itemize}
\item[(a).] Assumptions 1 and 2 above are the same as in \citet{eichler08} where the first requires existence of all order moments of $Y_t$ and the second ensures consistency of the estimated spectral matrix. It must be noted that the assumptions on the kernel and bandwidth are primarily for establishing asymptotic result in Theorem \ref{thm:f_ij_testing} and can be weakened for Theorems \ref{thm:r_ab_consistency} and \ref{thm:R_consistency}. 
\item[(b).] In the non-Gaussian setting, understanding tail behavior and verifying higher order moments is a non-trivial problem and has been explored in tail-index estimation (\citet{hill1975}, \citet{smooth_hill}) with an assumption on the tail distribution.
\end{itemize}
\end{remark}

\begin{theorem}\label{thm:r_ab_consistency}
Suppose that Assumptions 1,2 are satisfied. Then as $T\rightarrow \infty$, 
\begin{equation}
\widehat{r}_{i,a,b}  \;    \xrightarrow[]{P} \; \int_a^b \; \sum_{r,s=1}^d  \; g_{ii,rs}(\omega) \; \overline{ g_{ii,rs}(\omega)} \; d\omega,
\vspace{-0.2cm}
\end{equation}
where $g_{ii}(\omega)$ is the $d_i \times d_i$ spectral matrix of $Y_{i,t}$ and $\xrightarrow[]{P}$ denotes convergence in probability.
\end{theorem}

\begin{proof}
See Appendix for details of the proof.
\end{proof}

Next we take a special case wherein we wish to test for the equality of spectral matrices of same dimensions over an interval $0<a<b<\pi$. Let us assume $d_i=d_j$ for some $i \neq j \in \{1,2,\hdots , N \}$ and define 
\begin{equation}
G_{ij}(\omega) \; = \; \begin{bmatrix}
G_{11}(\omega) & G_{12}(\omega) \\ 
G_{21}(\omega) & G_{22}(\omega)
\end{bmatrix} = \begin{bmatrix}
g_{ii}(\omega) & g_{ij}(\omega) \\
g_{ji}(\omega) & g_{jj}(\omega)
\end{bmatrix}
\end{equation}
where the $d_i \times d_i$ matrices $G_{ab}$, for $a,b=1,2$, are the respective spectral and cross-spectral matrices of the processes $Y_{i,t}$ and $Y_{j,t}$. We consider testing    
\begin{equation}
H_0 \; : \; \ g_{ii} (\omega) = g_{jj}(\omega) \; \forall \; \omega \in (a,b)
\end{equation}  
where $0<a<b<\pi$ and $i,j= 1,2,...,N$ and $i \neq j$ . The test statistic is  

\begin{equation}
\widehat{D}_{i,j} = \int_{a}^{b} || vec( \hat{g}_{ii}(\omega) - \hat{g}_{jj}(\omega) )  ||_2^2 d\omega.
\end{equation}

\begin{theorem} \label{thm:f_ij_testing}
Suppose that Assumptions 1,2 are satisfied. Then as $T\rightarrow \infty$ under $H_0$ we have
\begin{equation}
 2 \pi T \sqrt{h}\; \hat{D}_{i,j}- \frac{\mu_{ij}}{\sqrt{h}}   \xrightarrow[]{D} N(0,\sigma^2_{ij})
\vspace{-0.2cm}
\end{equation}
where
\vspace{-0.2cm}
\begin{equation}\label{e:mu-f-ij}
\mu_{ij} = A_K \int_{-\pi}^{ \pi} 1_{\omega \in  (a,b)  } \Big(\; \sum_{p_1,p_2=1 }^{2} \big(\; -1 \; + \; 2 \delta_{p_1 p_2} \; \big) |tr( G_{p_1 p_2}(\omega) )|^2   \Big) \textrm{d}\omega
\vspace{-0.2cm}
\end{equation}
and
\vspace{-0.2cm}
\begin{equation} \label{e:var-f-ij}
\sigma_{ij}^2 = B_K \int_{-\pi}^{\pi} 1_{\omega \in  (a,b)  } \Big( \; \sum_{p_1,p_2,p_3,p_4=1}^{2} (\; -1 \; + \; 2\delta_{p_1 p_2} \;)\;(\; -1 \; + \; 2\delta_{p_3 p_4}  \;) | tr(\; G_{p_1 p_3}(\omega) \overline{G_{p_2 p_4}(\omega)}^T \; )|^2    \Big) \textrm{d}\omega.
\vspace{-0.2cm}
\end{equation}
where $\xrightarrow[]{D}$ denotes convergence in distribution,
 $A_K =  \int_{-\pi}^{ \pi } K^2(v) dv $, $B_K=4 \int_{a - \pi}^{b+ \pi}\; \Big( \int_{- \pi}^{  \pi} K(u)K(u+v)du \Big)^2\; dv$, $\delta_{rs}=I(r=s)$ is
 the Kronecker delta and tr($\cdot$) denotes the trace of a matrix.
\end{theorem}
\vspace{-0.2cm}
\begin{proof}
See Appendix for details of the proof.
\end{proof}

\noindent We now turn to the FS-ratio statistic $\widehat{R}_{i,a,b}$ defined in \eqref{e:R-def-samp}. It can be observed that this quantity can be written as 
\begin{gather} \label{e:R-def-samp-rewritten}
\widehat{R}_{i,a,b}= \frac{ 2 \int_a^b || vec( \widehat{g}_{ii}(\omega) ) ||_2^2 d \omega}{\int_{-\pi}^{ \pi} ||  vec( \widehat{g}_{ii}(\omega) ) ||_2^2 d \omega} 
= \frac{2 \int_a^b || vec( \widehat{g}_{ii}(\omega) ) ||_2^2 d \omega}{2 \int_a^b || vec( \widehat{g}_{ii}(\omega) ) ||_2^2 d \omega + \int_{\overline{\Pi}_{(a,b)}} || vec( \widehat{g}_{ii}(\omega) ) ||_2^2 d \omega} \notag \\
= \Big( 1 +   \frac{\int_{\overline{\Pi}_{(a,b)}} || vec( \widehat{g}_{ii}(\omega) ) ||_2^2 d \omega}{2 \int_a^b || vec( \widehat{g}_{ii}(\omega) ) ||_2^2 d \omega}   \Big)^{-1},
\end{gather}
where $\overline{\Pi}_{(a,b)} = \lbrack -\pi, \pi \rbrack \setminus (a,b) \cup (-b,-a)$ for some $0<a<b<\pi$ and $i=1,2, \hdots, N$ . Now we state the result that establishes consistency of $\widehat{R}_{i,a,b}$.

\begin{theorem}\label{thm:R_consistency}
Suppose that Assumptions 1,2 are satisfied and that for a given $0<a<b<\pi$, $r_{i,a,b}>0$ and $r_{i,\overline{\Pi}_{(a,b)}}>0$. Then as $T\rightarrow \infty$, 
\begin{equation}
\widehat{R}_{i,a,b}   \;  \xrightarrow[]{P} \; \Big( 1 + \frac{r_{i,\overline{\Pi}_{(a,b)}}}{2\; r_{i,a,b}} \Big)^{-1}
\vspace{-0.2cm}
\end{equation}
where $r_{i,a,b} = \int_a^b \;||g_{ii}(\omega)||^2 d\omega$  and $r_{i,\overline{\Pi}_{(a,b)}} = \int_{\overline{\Pi}_{(a,b)}} \;||g_{ii}(\omega)||^2 d\omega$.
\end{theorem}
\vspace{-0.2cm}
\begin{proof}
See Appendix for details of the proof.
\end{proof}  
 
\noindent Note that in finite sample situations explored in Sections \ref{s:simulation}, \ref{s:application}, we utilize the block bootstrap technique of \cite{politis94} for resampling from a stationary process. This is done to obtain sample quantiles of the FS-ratio statistic $\widehat{R}_{i,a,b}$.

\section{Simulation study}
\label{s:simulation}

In this section we illustrate the performance of the FS-ratio statistic in capturing spread of spectral information using simulated examples. We consider three simulation schemes and report the key summaries of the FS-ratio statistic across repetitions of each of the three schemes. In addition, 95\% bootstrap confidence limits for the FS-ratio statistic is computed from $B=500$ bootstrap replications. Here we utilize the block bootstrap procedure of \cite{politis04,patton09}. For an estimate of the spectral matrix defined in \eqref{eq:estimated_spectral_matrix}, the Bartlett-Priestley kernel with bandwidth $h=N^{-0.4}$ and the Daniell kernel, see Example 10.4.1 in \cite{bd91} with $m=\sqrt{N}$ were implemented. Similar results were obtained for the two kernel choices and only the results from the latter are presented. 

\vspace{0.5cm}

\noindent \textbf{\underline{Scheme 1}}: We simulate the $p_i$-variate process $Y_{i,t}  = (Y_{1,i,t},Y_{1,2,t},\hdots,Y_{p_i,i,t})^{'}$ where each $Y_{k,i,t}$ are independently generated univariate stationary AR(2) process given by 

$$ Y_{k,i,t} = \phi_{i,1} Y_{k,i,t-1} + \phi_{i,2} Y_{k,i,t-2} + \epsilon_{k,i,t} $$

\noindent $\phi_{i,1} = 2 \xi_{i} \cos(\theta_i) $, $\phi_{i,2} = -\xi_{i}^2$, $\epsilon_{k,i,t}$ are i.i.d $N(0,1)$ and $k=1,2,\hdots, p_i$,  $i=1,2,\hdots,N=500$, $t=1,2,\hdots,T=1000$. The dimension $p_i$ for $Y_{i,t}$ is randomly chosen from $\{2,3, \hdots,30 \}$. Here $\xi_i \sim U(0.8,0.98)$ and $\theta_i$ is given by 

$$ \theta_i =
\left\{
	\begin{array}{ll}
		\cos(\frac{ 4 \pi}{25})  & \mbox{if } i < \frac{N}{2} \\
		\cos(\frac{4 \pi}{5}) & \mbox{if } i \geq \frac{N}{2}
	\end{array}
\right. $$

\noindent Figure \ref{fig_scheme1_realization} presents a sample illustration of a bivariate realization from Scheme 1. The plot includes two components from one epoch $i<N/2$ and another epoch $i \geq N/2$. Similar illustrations for Schemes 2 and 3 below can be found in Figures \ref{fig_scheme2_realization}, \ref{fig_scheme3_realization} respectively.  

\begin{figure}[H]
\centering
\includegraphics[scale=0.45]{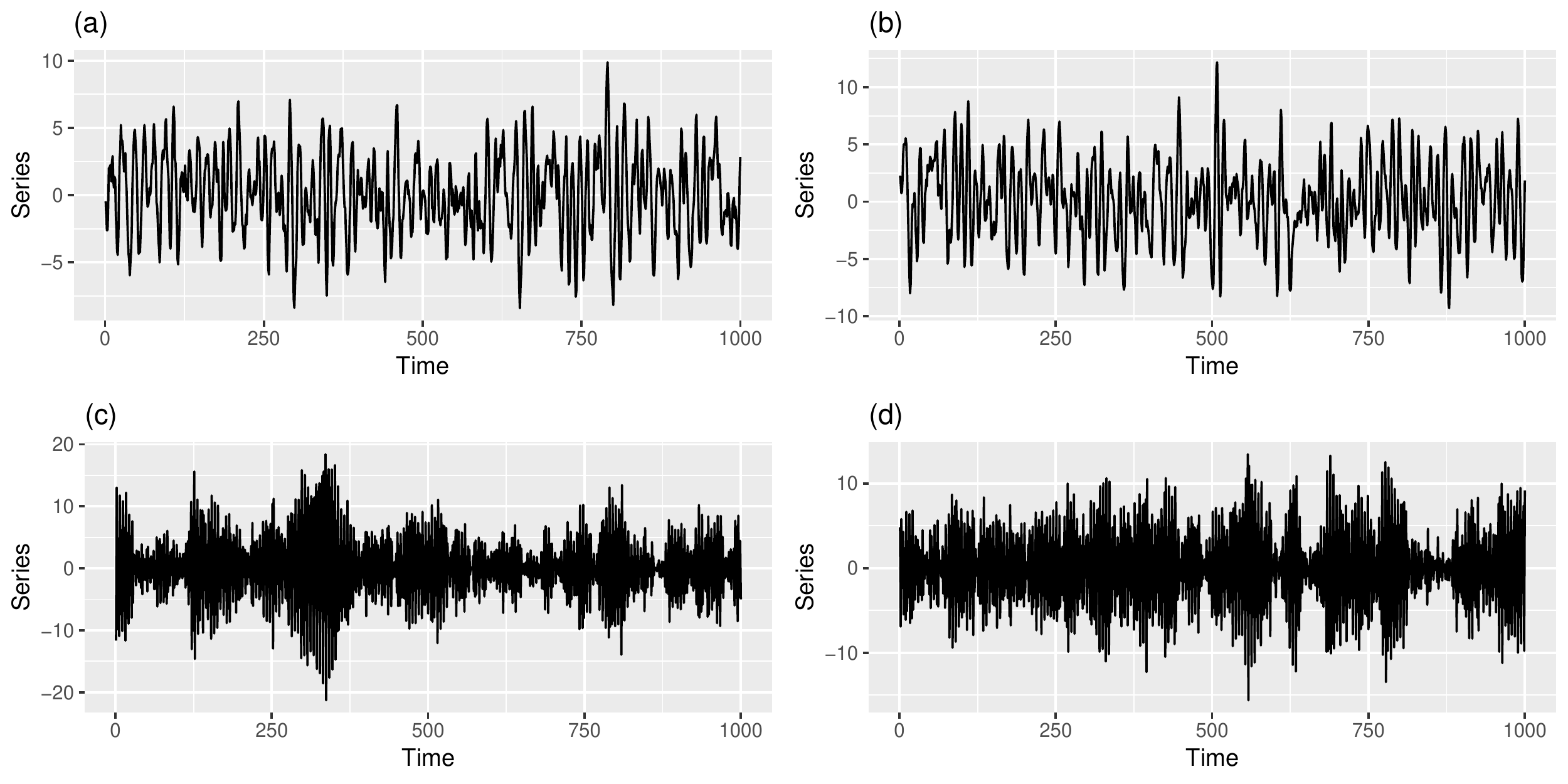}
\caption{Sample bivariate realization from \textbf{Scheme 1}: (a) and (b) are two components of $Y_{i,t} \in \mathbb{R}^2$ when $i<N/2$. (c) and (d) are two components of $Y_{i,t} \in \mathbb{R}^2$ when $i \geq N/2$.} \label{fig_scheme1_realization}
\end{figure}

Tables \ref{tab:sim_scheme_1_first}, \ref{tab:sim_scheme_1_second} contain the numerical summaries of the FS-ratio statistic over 100 replications of Scheme 1. Note that the phase parameter $\theta_i$ for $i<N/2$ in Scheme 1 is at $4 \pi /25$ on a $(0,\pi)$ scale or equivalently at 0.0796 on a $(0,0.5)$ scale. We see from Table \ref{tab:sim_scheme_1_first} that almost all of the spectral information is contained in the first two chosen frequency ranges around this peak. Similarly for $i \geq N/2$, the phase parameter is at $4 \pi /5$ on a $(0,\pi)$ scale or equivalently at 0.3981 on a $(0,0.5)$ scale. Figure \ref{fig:scheme1_histogram} plots a histogram density of the FS-ratio statistic from the 100 replications and similar histogram densities for Schemes 2 and 3 can be found in Figures \ref{fig:scheme2_histogram}, \ref{fig:scheme3_histogram}. From Table \ref{tab:sim_scheme_1_second} we notice that the last two chosen frequency ranges have all of the spectral information. 

\begin{table}[h]
\begin{center}
\begin{tabular}{|c|c|c|c|c|c|}
\hline
Frequency Range  & Mean & Median & SD & Lower & Upper \\
(a,b) &     &    & & CI & CI    \\        
\hline
(0,0.08) & 0.5342 & 0.5391 & 0.0221 & 0.4984 & 0.6253  \\
\hline
(0.08,0.16) & 0.4486 & 0.4544 & 0.0227 & 0.3566 & 0.4831 \\
\hline
(0.16,0.24) & 0.0002 & 0.0002 & 0.0001 & 0.0005 & 0.0019 \\
\hline
(0.24,0.32)  &  0 & 0 & 0 & 0 & 0.0002 \\
\hline
(0.32,0.40)  &  0 & 0 & 0 & 0 & 0 \\
\hline
(0.40,0.48)  &  0 & 0 & 0 & 0 & 0 \\
\hline
\end{tabular}
\end{center}
\vspace{-0.3cm}
\caption{\textbf{Scheme 1, epochs 1-249}: Numerical summaries of FS-ratio statistic $\widehat{R}_{i,a,b}$ for  epochs $i=1,2,\hdots,  249$ for specified frequency ranges $(a,b)$. Here $(a,b) \subset  (0,0.5)$ and $(0,0.5)$ corresponds to the interval $(0,\pi)$.} \label{tab:sim_scheme_1_first}
\end{table}

\begin{table}[h]
\begin{center}
\begin{tabular}{|c|c|c|c|c|c|}
\hline
Frequency Range  & Mean & Median & SD & Lower & Upper \\
(a,b) &     &    & & CI & CI    \\        
\hline
(0,0.08) & 0 & 0 & 0 & 0 & 0  \\
\hline
(0.08,0.16) & 0 & 0 & 0 & 0 & 0 \\
\hline
(0.16,0.24) & 0 & 0 & 0 & 0 & 0 \\
\hline
(0.24,0.32)  & 0.0003 & 0.0002 & 0.0002 & 0.0005 & 0.0017 \\
\hline
(0.32,0.40)  & 0.4561& 0.4595 & 0.0181 & 0.3786 & 0.4903 \\
\hline
(0.40,0.48)  & 0.5205 & 0.5210 & 0.0169 & 0.4759 & 0.5826 \\
\hline
\end{tabular}
\end{center}
\vspace{-0.3cm}
\caption{\textbf{Scheme 1, epochs 250-500}: Numerical summaries of FS-ratio statistic $\widehat{R}_{i,a,b}$ for  epochs $i=250,2,\hdots,  500$ for specified frequency ranges $(a,b)$. Here $(a,b) \subset  (0,0.5)$ and $(0,0.5)$ corresponds to the interval $(0,\pi)$.} \label{tab:sim_scheme_1_second}
\end{table}

\begin{figure}[H]
\centering
\includegraphics[scale=0.45]{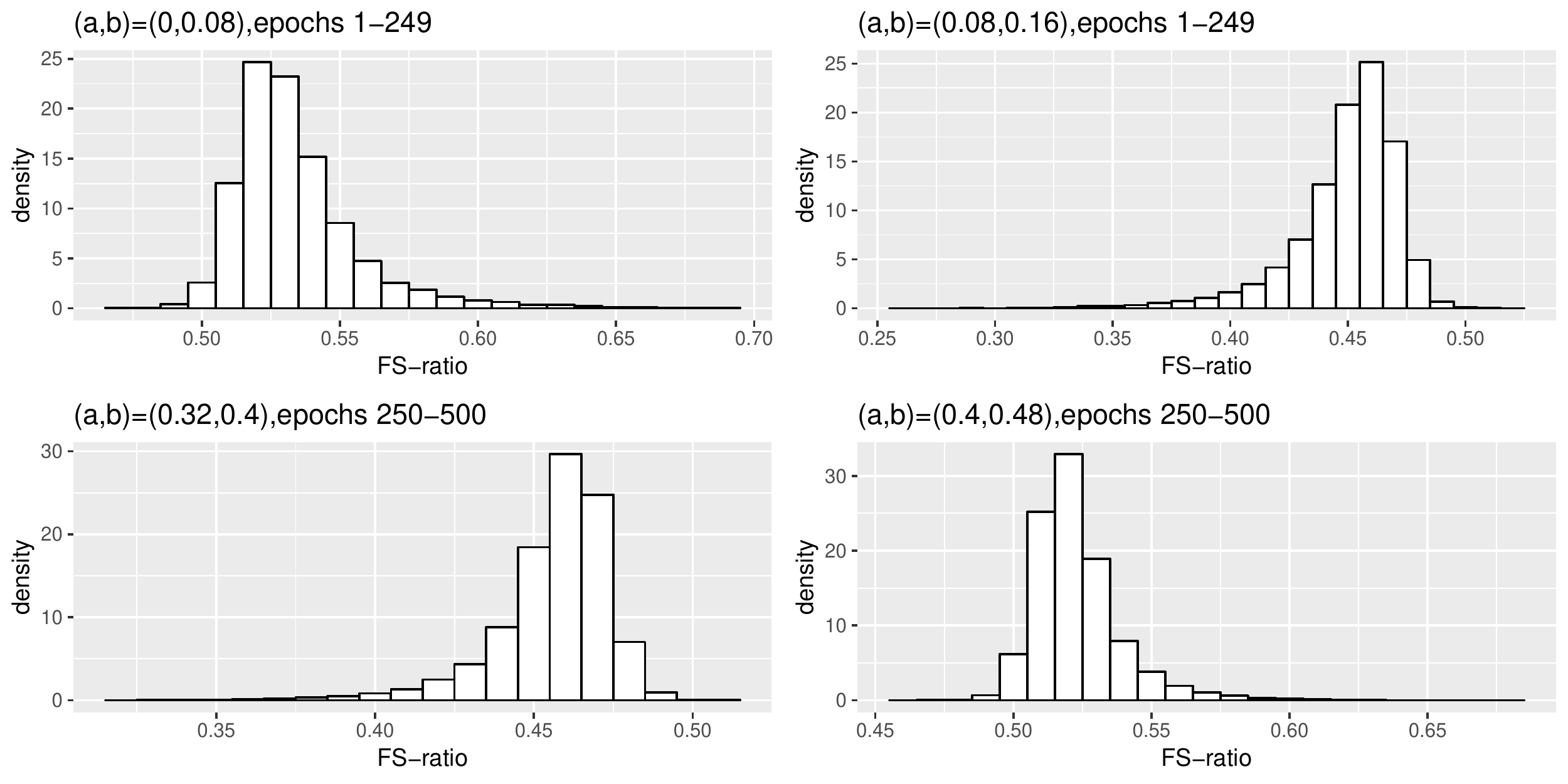}
\caption{\textbf{Scheme 1}: Histogram density of the FS-ratio statistic for different frequency ranges $(a,b) \subset (0,0.5)$.} \label{fig:scheme1_histogram}
\end{figure}

\vspace{0.5cm}

\noindent \textbf{\underline{Scheme 2}}: Similar to Scheme 1, we simulate the $p_i$-variate process $Y_{i,t}  = (Y_{1,i,t},Y_{1,2,t},\hdots,Y_{p_i,i,t})^{'}$ where each $Y_{k,i,t}$ are independently generated univariate stationary AR(2) process given by 

$$ Y_{k,i,t} = \phi_{i,1} Y_{k,i,t-1} + \phi_{i,2} Y_{k,i,t-2} + \epsilon_{k,i,t} $$

\noindent $\phi_{i,1} = 2 \xi_{i} \cos(\theta_i) $, $\phi_{i,2} = -\xi_{i}^2$. The $p_i \times p_i$ variance matrix of the Gaussian noise $\epsilon_{i,t}$ is given by

$$ V(\epsilon_{i,t}) = \begin{bmatrix}
1 & \rho & \rho^2 & \hdots & \rho^{p_i-1} \\
\rho & 1 & \rho  & \hdots & \rho^{p_i-2} \\
\vdots \\
\rho^{p_i-1} & \rho^{p_i-2} & \rho^{p_i-3} & \hdots & 1
\end{bmatrix} $$
$\rho = 0.4$ and $k=1,2,\hdots, p_i$,  $i=1,2,\hdots,N=500$, $t=1,2,\hdots,T=1000$. The dimension $p_i$ for $Y_{i,t}$ is randomly chosen from $\{2,3, \hdots,30 \}$. Here again, $\xi_i \sim U(0.8,0.98)$ and $\theta_i$ is given by 

$$ \theta_i =
\left\{
	\begin{array}{ll}
		\cos(\frac{ 4 \pi}{25})  & \mbox{if } i < \frac{N}{2} \\
		\cos(\frac{4 \pi}{5}) & \mbox{if } i \geq \frac{N}{2}
	\end{array}
\right. $$

\begin{figure}[H]
\centering
\includegraphics[scale=0.45]{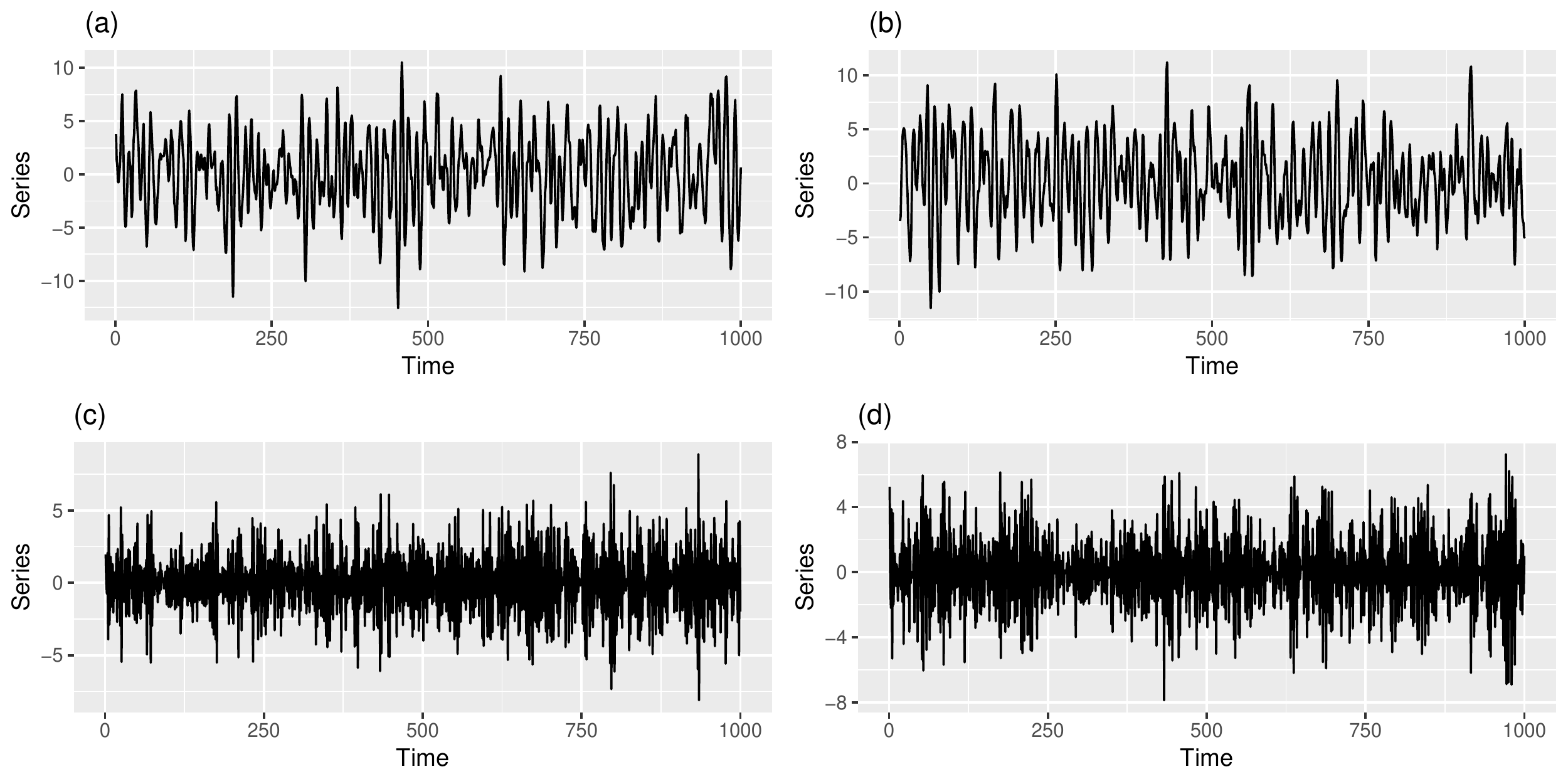}
\caption{Sample bivariate realization from \textbf{Scheme 2}: (a) and (b) are two components of $Y_{i,t} \in \mathbb{R}^2$ when $i<N/2$. (c) and (d) are two components of $Y_{i,t} \in \mathbb{R}^2$ when $i \geq N/2$.} \label{fig_scheme2_realization}
\end{figure}

Tables \ref{tab:sim_scheme_2_first}, \ref{tab:sim_scheme_2_second} contain the numerical summaries of the FS-ratio statistic over 100 replications of the model in Scheme 2. As in Scheme 1, the phase parameter $\theta_i$ for $i<N/2$ is at 0.0796 on a $(0,0.5)$ scale for $i<N/2$ and at 0.3981 on a $(0,0.5)$ scale for $i \geq N/2$. As in Scheme 1, results from Table \ref{tab:sim_scheme_2_first} indicate most of the spectral information are present in the first two chosen frequency ranges. Similarly for $i \geq N/2$, Table \ref{tab:sim_scheme_2_second} shows that the last two chosen frequency ranges have all of the spectral information.

\begin{table}[H]
\begin{center}
\begin{tabular}{|c|c|c|c|c|c|}
\hline
Frequency Range  & Mean & Median & SD & Lower & Upper \\
(a,b) &     &    & & CI & CI    \\        
\hline
(0,0.08) & 0.5371 & 0.5327 & 0.0238 & 0.4977 & 0.6284  \\
\hline
(0.08,0.16) & 0.4459 & 0.4504 & 0.0239 & 0.3549 & 0.4843 \\
\hline
(0.16,0.24) & 0.0002 & 0.0002 & 0.0001 & 0.0005 & 0.0020 \\
\hline
(0.24,0.32)  &  0 & 0 & 0 & 0 & 0.0002 \\
\hline
(0.32,0.40)  &  0 & 0 & 0 & 0 & 0 \\
\hline
(0.40,0.48)  &  0 & 0 & 0 & 0 & 0 \\
\hline
\end{tabular}
\end{center}
\vspace{-0.3cm}
\caption{\textbf{Scheme 2, epochs 1-249}: Numerical summaries of FS-ratio statistic $\widehat{R}_{i,a,b}$ for  epochs $i=1,2,\hdots,  249$ for specified frequency ranges $(a,b)$. Here $(a,b) \subset  (0,0.5)$ and $(0,0.5)$ corresponds to the interval $(0,\pi)$.} \label{tab:sim_scheme_2_first}
\end{table}

\begin{table}[H]
\begin{center}
\begin{tabular}{|c|c|c|c|c|c|}
\hline
Frequency Range  & Mean & Median & SD & Lower & Upper \\
(a,b) &     &    & & CI & CI    \\        
\hline
(0,0.08) & 0 & 0 & 0 & 0 & 0  \\
\hline
(0.08,0.16) & 0 & 0 & 0 & 0 & 0 \\
\hline
(0.16,0.24) & 0 & 0 & 0 & 0 & 0.0001 \\
\hline
(0.24,0.32)  & 0.0003 & 0.0003 & 0.0002 & 0.0005 & 0.0018 \\
\hline
(0.32,0.40)  & 0.4531 &  0.4566 & 0.0196 & 0.3758 & 0.4907 \\
\hline
(0.40,0.48)  & 0.5252 & 0.5225 & 0.0172 & 0.4810 & 0.5948 \\
\hline
\end{tabular}
\end{center}
\vspace{-0.3cm}
\caption{\textbf{Scheme 2, epochs 250-500}: Numerical summaries of FS-ratio statistic $\widehat{R}_{i,a,b}$ for  epochs $i=250,2,\hdots,  500$ for specified frequency ranges $(a,b)$. Here $(a,b) \subset  (0,0.5)$ and $(0,0.5)$ corresponds to the interval $(0,\pi)$.} \label{tab:sim_scheme_2_second}
\end{table}

\begin{figure}[H]
\centering
\includegraphics[scale=0.45]{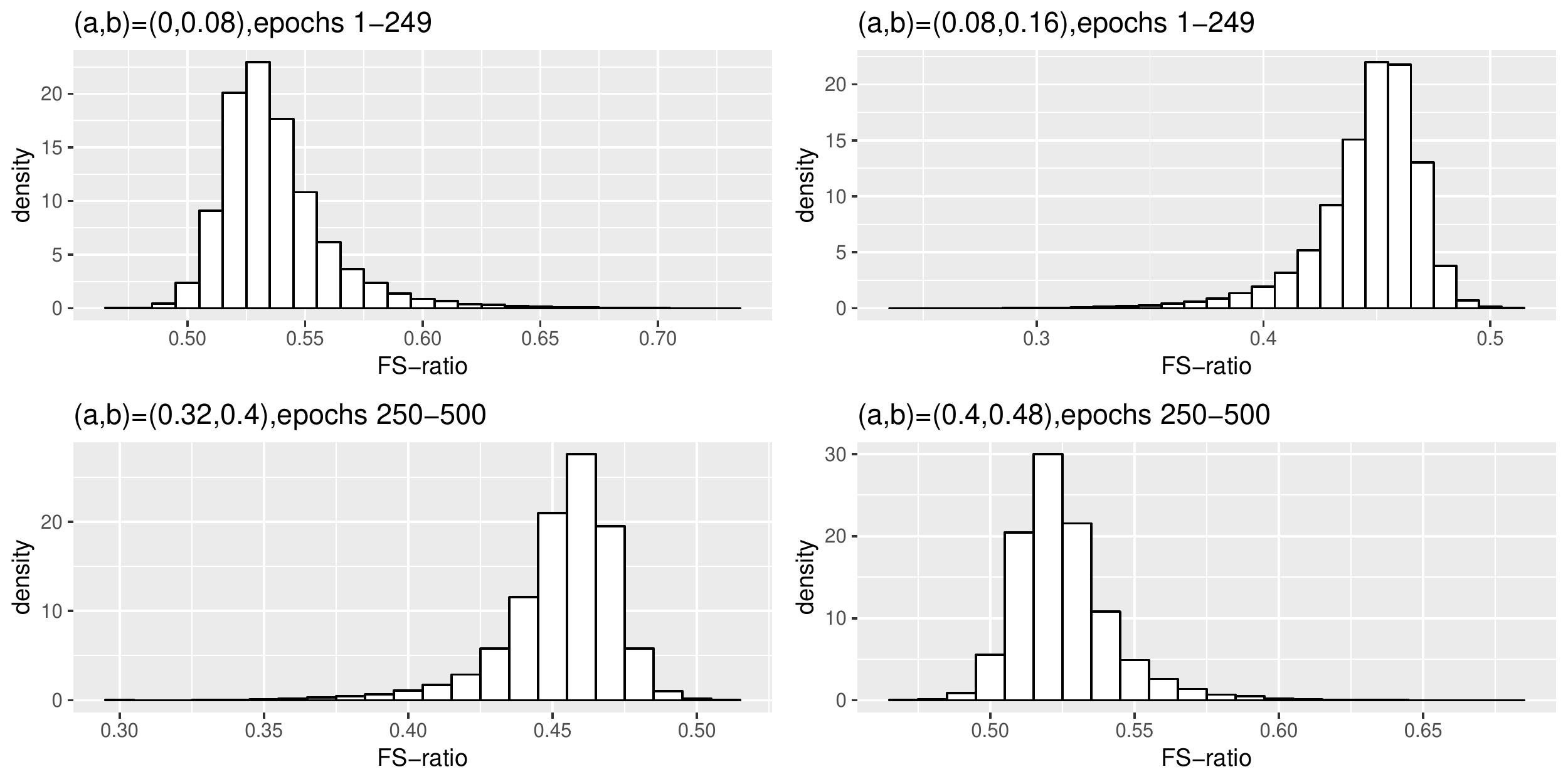}
\caption{\textbf{Scheme 2}: Histogram density  of the FS-ratio statistic for different frequency ranges $(a,b) \subset (0,0.5)$.} \label{fig:scheme2_histogram}
\end{figure}

\vspace{0.5cm}

\noindent \textbf{\underline{Scheme 3}}: Here we let $p=30$ and follow Scheme 1 in generating a $p$-variate process $Y_{i,t}  = (Y_{1,i,t},Y_{1,2,t},\hdots,Y_{p,i,t})^{'}$ for $i=1,2,\hdots,N=500$ and $t=1,2,\hdots=1000$. Then we obtain $X_{i,t} = A_iY_{i,t}$ where $A_i = 1_{(i<\frac{N}{2})} \; I_{p_i}A_1  + 1_{(i \geq \frac{N}{2})} \; I_{p_i}A_2$ and $A_1$ and $A_2$ are two $p \times p$ randomly generated orthogonal matrices and $I_{p_i}$ is the $p_i \times p_i$ identity matrix. We consider $X_{i,t} \in \mathbb{R}^{p_i}$ and study the spread of spectral properties across the $N=500$ epochs. 

\begin{figure}[H]
\centering
\includegraphics[scale=0.45]{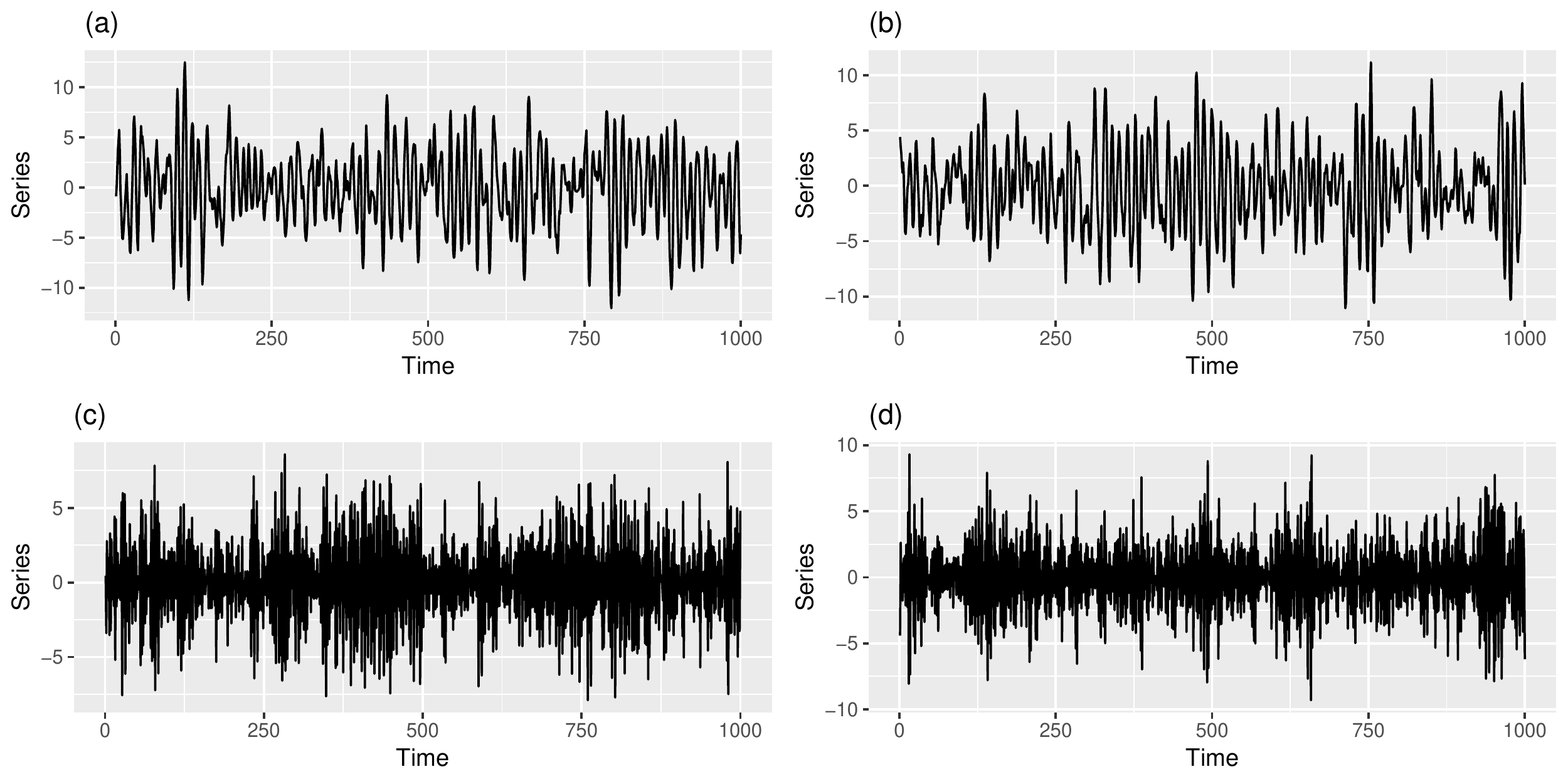}
\caption{Sample bivariate realization from \textbf{Scheme 3}: (a) and (b) are two components of $Y_{i,t} \in \mathbb{R}^2$ when $i<N/2$. (c) and (d) are two components of $Y_{i,t} \in \mathbb{R}^2$ when $i \geq N/2$.} \label{fig_scheme3_realization}
\end{figure}

Tables \ref{tab:sim_scheme_3_first}, \ref{tab:sim_scheme_3_second} contain the numerical summaries of the FS-ratio statistic over 100 replications of the model in Scheme 3. Here we look at $X_{i,t} = A_iY_{i,t}$ which is a mixture of the components of $Y_{i,t}$ generated as in Scheme 1. Note that the peak of the spectral densities of the components of $X_{i,t}$ is still at the phase parameter $\theta_i$ defined in Scheme 1. Hence, the results from Table \ref{tab:sim_scheme_3_first}, \ref{tab:sim_scheme_3_second} are similar to the results from Scheme 1. 

\begin{table}[H]
\begin{center}
\begin{tabular}{|c|c|c|c|c|c|}
\hline
Frequency Range  & Mean & Median & SD & Lower & Upper \\
(a,b) &     &    & & CI & CI    \\        
\hline
(0,0.08) & 0.5342 & 0.5321 & 0.0155 & 0.4871 & 0.5955  \\
\hline
(0.08,0.16) & 0.4489 & 0.4510 & 0.0159 & 0.3872 & 0.4949 \\
\hline
(0.16,0.24) & 0.0002 & 0.0002 & 0.0001 & 0.0003 & 0.0011 \\
\hline
(0.24,0.32)  &  0.0001 & 0.0001 & 0 & 0 & 0.0001 \\
\hline
(0.32,0.40)  &  0 & 0 & 0 & 0 & 0 \\
\hline
(0.40,0.48)  &  0 & 0 & 0 & 0 & 0 \\
\hline
\end{tabular}
\end{center}
\vspace{-0.3cm}
\caption{\textbf{Scheme 3, epochs 1-249}: Numerical summaries of FS-ratio statistic $\widehat{R}_{i,a,b}$ for  epochs $i=1,2,\hdots,  249$ for specified frequency ranges $(a,b)$. Here $(a,b) \subset  (0,0.5)$ and $(0,0.5)$ corresponds to the interval $(0,\pi)$.} \label{tab:sim_scheme_3_first}
\end{table}

\begin{table}[H]
\begin{center}
\begin{tabular}{|c|c|c|c|c|c|}
\hline
Frequency Range  & Mean & Median & SD & Lower & Upper \\
(a,b) &     &    & & CI & CI    \\        
\hline
(0,0.08) & 0 & 0 & 0 & 0 & 0  \\
\hline
(0.08,0.16) & 0 & 0 & 0 & 0 & 0 \\
\hline
(0.16,0.24) & 0 & 0 & 0 & 0 & 0.0001 \\
\hline
(0.24,0.32)  & 0.0003 & 0.0003 & 0.0001 & 0.0003 & 0.0011 \\
\hline
(0.32,0.40)  & 0.4553 &  0.4570 & 0.0130 & 0.4021 & 0.4987 \\
\hline
(0.40,0.48)  & 0.5234  & 0.5219 & 0.0119 & 0.4782  & 0.5738  \\
\hline
\end{tabular}
\end{center}
\vspace{-0.3cm}
\caption{\textbf{Scheme 3, epochs 250-500}: Numerical summaries of FS-ratio statistic $\widehat{R}_{i,a,b}$ for  epochs $i=250,2,\hdots,  500$ for specified frequency ranges $(a,b)$. Here $(a,b) \subset  (0,0.5)$ and $(0,0.5)$ corresponds to the interval $(0,\pi)$.} \label{tab:sim_scheme_3_second}
\end{table}

\begin{figure}[H]
\centering
\includegraphics[scale=0.45]{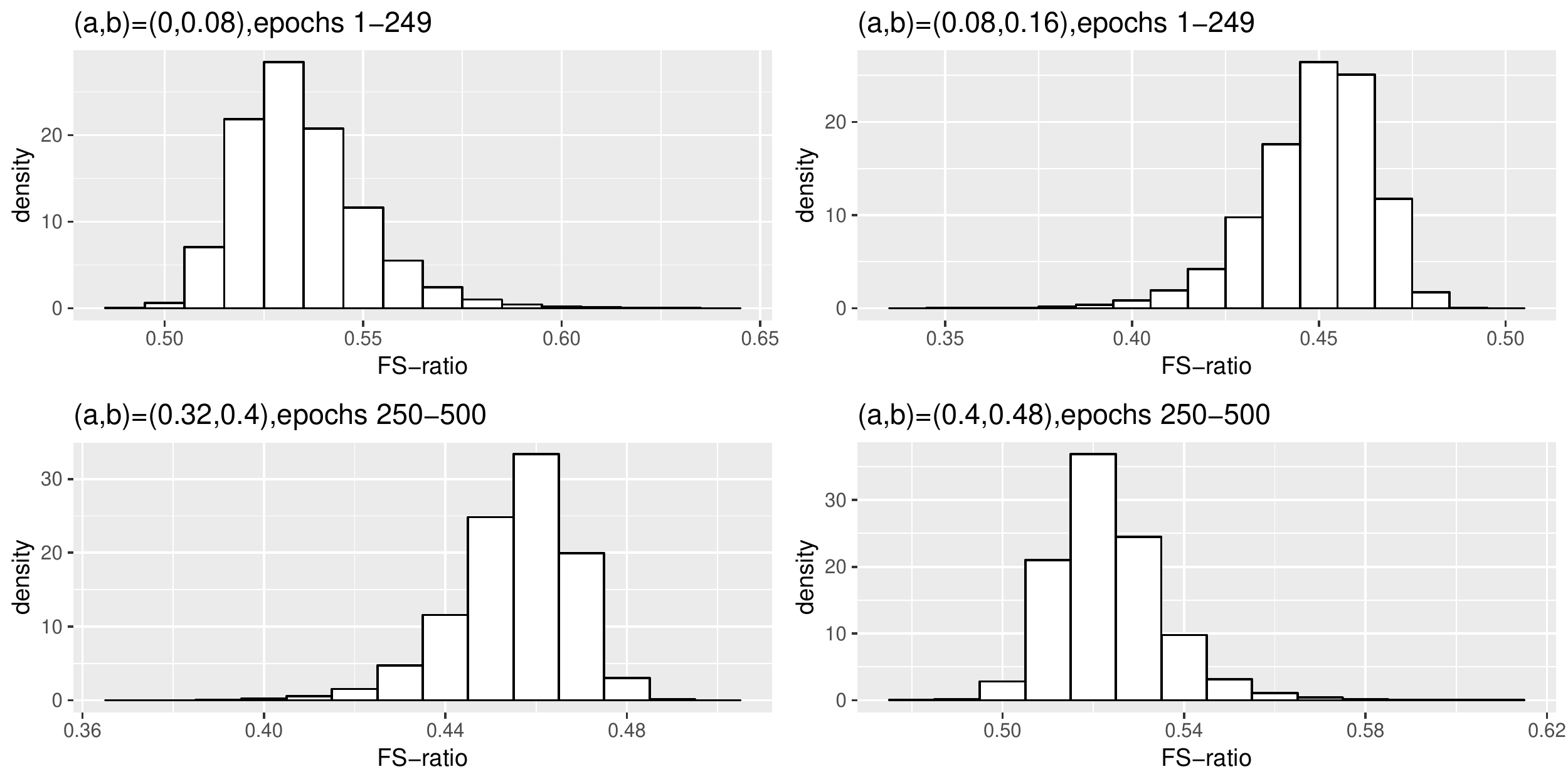}
\caption{\textbf{Scheme 3}: Histogram density of the FS-ratio statistic for different frequency ranges $(a,b) \subset (0,0.5)$.} \label{fig:scheme3_histogram}
\end{figure}

\section{Analysis of complexity of rat local field potentials in a stroke experiment}
\label{s:application}

In this section we investigate the ability of the FS-ratio to identify changes in the spectral properties of the 
local field potential (LFP) of a rat.\footnote{Local field potential data on the experimental rat comes from the stroke experiment conducted at Frostig laboratory at University of California Irvine: \texttt{http://frostiglab.bio.uci.edu/Home.html}. } The aim is to identify changes in complexity 
and structure of the multivariate cortex signal over the course of the experiment. It is also of interest 
to understand the differential roles of frequency bands and determine the specific bands that 
demonstrate the most significant changes that occurred due to the stroke. 

At 32 locations on the rat's cortex, microelectrodes are inserted: 4 layers in the cortex,
at 300$\mu m$, 700$\mu m$, 1100$\mu m$ and 1500$\mu m$ and 8 microelectodes lined
up in each of the 4 layers. We look at the field potential specific to the 32 locations recorded for a total duration of 10 minutes. This signal is divided into 600 epochs with each epoch comprising of 1 second worth of data. The sampling rate here is 1000 Hz resulting in $T=1000$ observations per epoch. Midway through the recording 
period (after epoch 300) a stroke is artificially induced by clamping the medial cerebral artery that supplied blood to the recorded area.  

As a first step in our analysis, we applied a component-wise univariate test of second-order stationarity (\citet{dwivedi_subbarao}) of the LFP signal at each epoch and presented the p-values from the tests in Figure \ref{fig:pvalue_component_plot}. The univariate tests indicate that the signal, within many epochs, is nonstationary. 

Next we model the observed 32-dimensional signal as a multivariate nonstationary time series using the SSA setup. We assume the observed 32 dimensional LFP signal $X_t$ is linearly generated by stationary and nonstationary sources in the cortex. More precisely we have,
\begin{equation}\label{e:ssa_assumption_on_lfp}
X_{i,t} \; = \; A_i Y_{i,t} \; + \; Z_{i,t}, \;\; i=1,2,\hdots,N=600, 
\end{equation} 
where $Y_{i,t} \in \mathbb{R}^{d_i}$ is latent stationary source, $A_i$ is a $p \times d_i$ unknown demixing matrix, $Z_{i,t}$ are the nonstationary sources. 

The next goal in the data analysis is to estimate the epoch-evolving dimension $d_i$ and the latent stationary 
time series $Y_{i,t} \in \mathbb{R}^{d_i}$ where $d_i < p$. This problem of starting with an 
observed nonstationary time series and, after some transformation, getting to a lower dimensional 
stationary time series has interesting applications in neuroscience. For instance, EEG signals measuring 
brain activity appear often as a multivariate nonstationary time series; 
see \citet{slex_ombao}, \citet{srinivasan_2003}, \citet{srinivasan_2006}, \citet{SSAbci}, 
\citet{wu_2016},  \citet{gao_2018}, \citet{euan_2019}  for examples. \citet{eeg_nonstationary_kaplan} regard the nonstationarity as background activity in the brain signal  
and removing this nonstationarity was seen to improve prediction accuracy in  neuroscience 
experiments; \citet{ssa09} and \citet{SSAbci}.  Thus the aim of SSA is to separate the stationary from the nonstationary sources within each 
epoch and focus our attention on the stationary sources. From a stroke neuroscientist's 
perspective, the stationary sources within a short epoch of 1 second are considered as the ``stable" components of the signal since they are consistent within that short interval.  Of course 
the transient components (nonstationary components) may also be of interest in other applications. 
In this section, we will demonstrate that subsequent analyses based on the stationary components 
can be powerful for some types of data.

The evolutionary dimension $d_i$ of the latent stationary sources were presented in Figure \ref{fig:lfp_dimension_plot}. The plot indicates increase in the number of stationary sources in post-stroke epochs (after epoch 300) and this agrees with the results in 
Figure \ref{fig:pvalue_component_plot} wherein more epochs after the stroke witness stationary behavior 
in the individual LFP components. It is indeed interesting that immediately post-occlusion (or immediately 
after stroke onset), the LFPs are highly synchronized: the plots of the observed LFP $X_{i,t}$ and the estimated  squared coherence between the 32 components (Figure \ref{fig:coherence_observed_lfp}) suggest that  different  electrodes look very similar and there is high coherence in between the entire network of electrodes 
at various frequency bands.  This was confirmed by the neuroscientists and recorded in her 
PhD dissertation (\citet{wann_thesis}). Our proposed method produced results that support the previous findings on coherence but it gave
an additional insight about the role of the stationary components that significantly explain this high degree of synchronicity. Next, we investigate further into the lead-lag cross-dependence between microelectrodes. We prewhitened the observed time series to make the lag-0 covariance matrix identity. More precisely, one considers $\Sigma^{-1/2}_i X_t$ where $\Sigma^{-1/2}_i$ is the inverse square root of the lag-0 covariance matrix $V(X_{i,t})$. We observe, in Figure \ref{fig:coherence_observed_lfp}, the significant drop in the magnitude of squared coherence after pre-whitening indicating that the dependence among the 32 components is predominantly due to a contemporaneous (i.e., lag-0) dependence. One can also notice, from the right plot in Figure \ref{fig:coherence_observed_lfp}, a drop in the coherence in the gamma frequency band after the stroke. 

\begin{figure}[H] 
\begin{subfigure}{0.5\textwidth}
\centering
\includegraphics[scale=0.3]{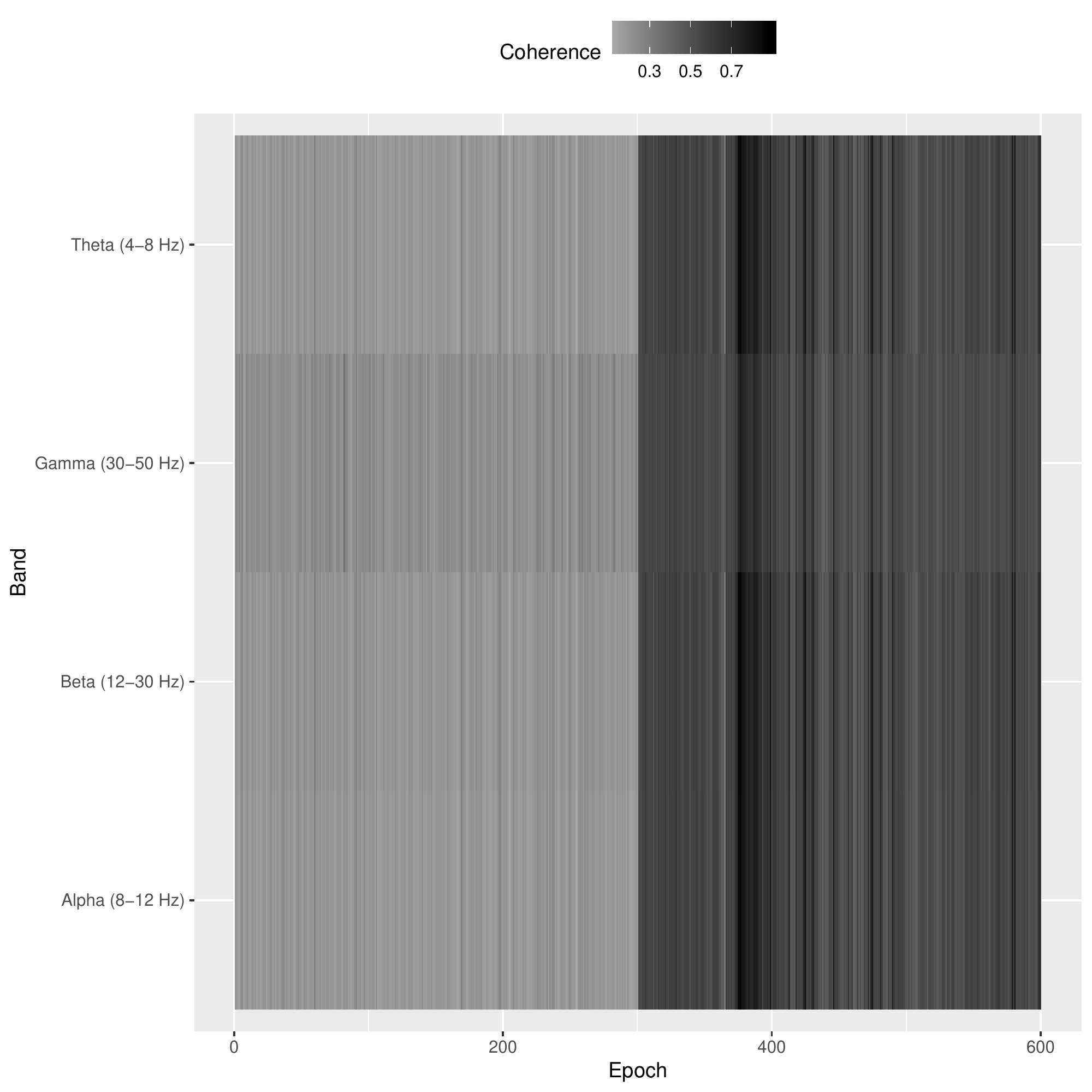}
\end{subfigure}
~
\begin{subfigure}{0.5\textwidth}
\centering
\includegraphics[scale=0.3]{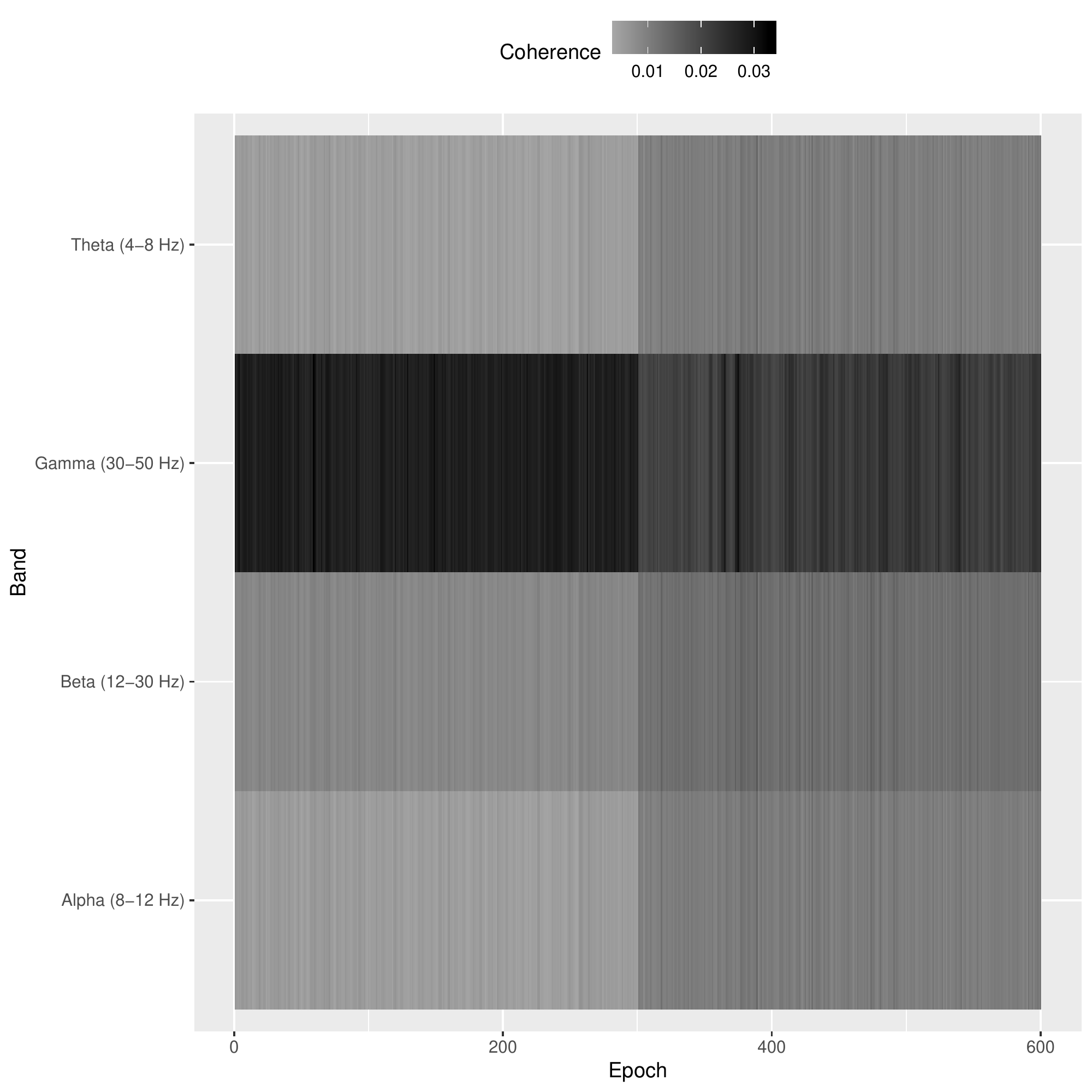}
\end{subfigure}
\caption{ \textbf{Left}: average squared coherency among the 32 components of the observed LFP signal across 600 epochs. The averages are computed across the specified frequency bands. \textbf{Right}: average squared coherency among the 32 components of the pre-whitened LFP signal across 600 epochs.} \label{fig:coherence_observed_lfp}
\end{figure}

We then estimated the latent stationary sources $Y_{i,t}$ for the $i=1,2,\hdots,N=600$ epochs using the DSSA method in \citet{sundararajan:2017}. In order to overcome identifiability issues in the model in \eqref{e:ssa_assumption_on_lfp}, SSA and PCA methods for time series assume an identity lag-0 covariance matrix for $X_{i,t}$ and resort to a prewhitening technique to achieve this. Figure \ref{fig:coherence_stationary} plots the average squared coherence in the non pre-whitened and prewhitened stationary sources across different frequency bands. Similar to the coherence pattern in the observed LFP in Figure \ref{fig:coherence_observed_lfp}, the left plot in Figure \ref{fig:coherence_stationary} witnesses an increase in the coherence after the occurrence of the stroke. Also, the right plot in Figure \ref{fig:coherence_stationary} indicates a substantial drop in the magnitude of coherence in the stationary sources. The prewhitened stationary sources have lower coherence than the coherence of the stationary 
sources based on the non-prewhitened. As noted, previous findings have already indicated an increased coherence 
post stroke onset. Our analysis provided an additional insight that the increase in the coherence post-stroke 
is due only to contemporaneous (or lag-0) dependence. This indicates perfect temporal synchrony in a sense 
that there is no lead-lag cross-dependence between the electrodes. This was suggested by visual inspection 
of the LFP traces and hypothesized by neuroscientists though never formally confirmed until now with 
our analysis.

\begin{figure}[H] 
\begin{subfigure}{0.5\textwidth}
\centering
\includegraphics[scale=0.3]{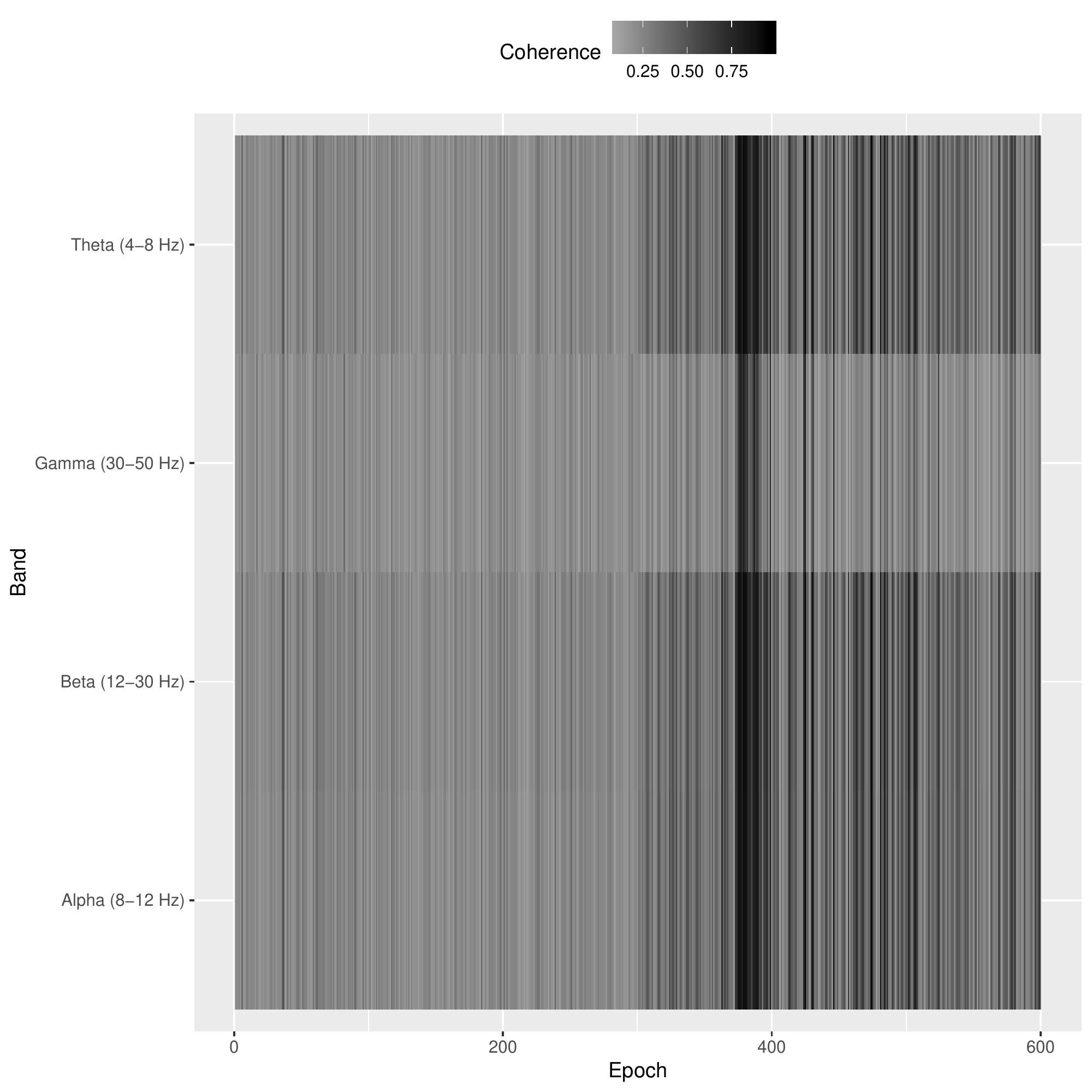}
\end{subfigure}
~
\begin{subfigure}{0.5\textwidth}
\centering
\includegraphics[scale=0.3]{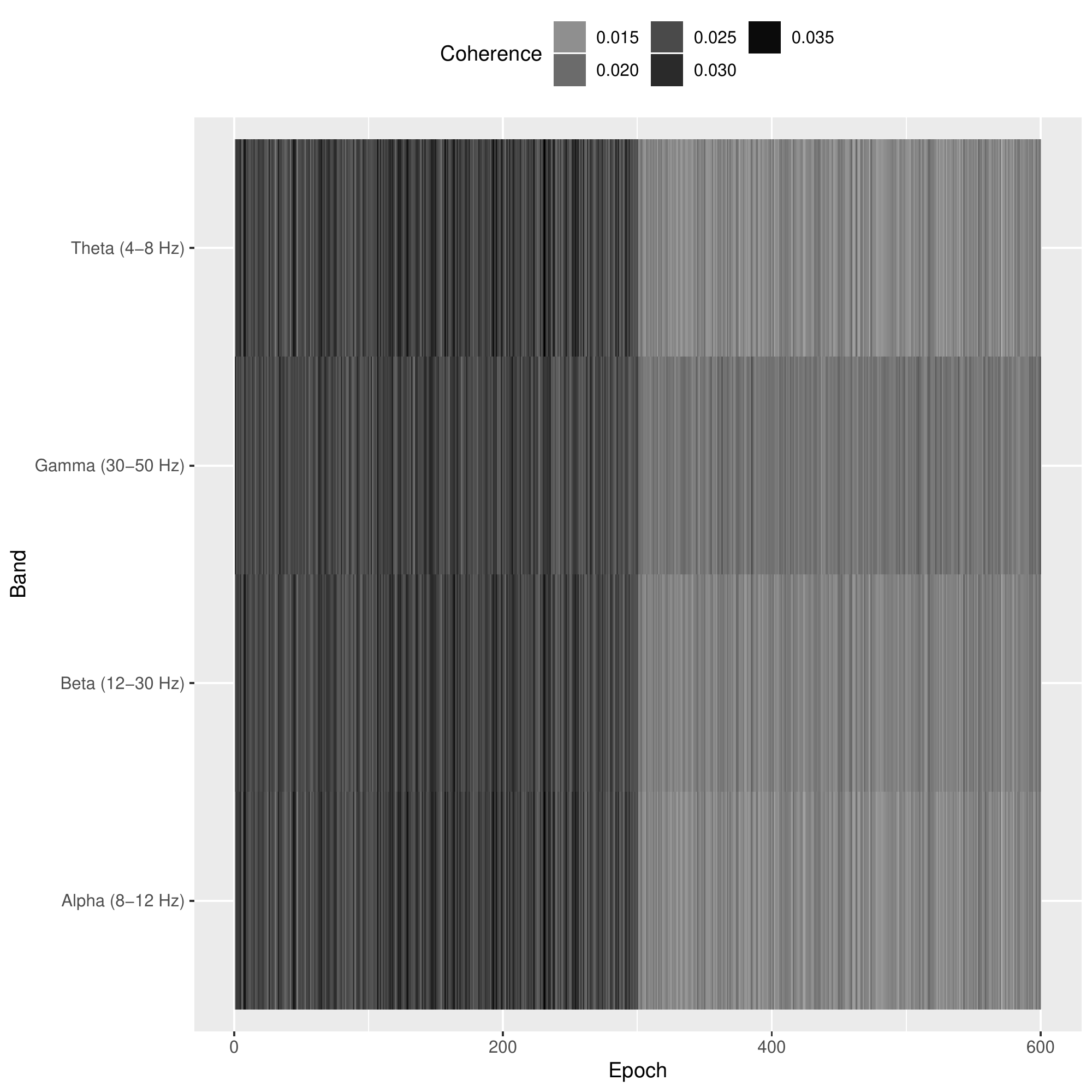}
\end{subfigure}
\caption{ \textbf{Left}: average squared coherency in the estimated stationary sources across 600 epochs. The averages are computed across the specified frequency bands. \textbf{Right}: average squared coherency  in the pre-whitened stationary sources across 600 epochs.} \label{fig:coherence_stationary}
\end{figure}

Next, the FS-ratio statistic was evaluated on these estimated stationary sources at each of the 600 epochs at various frequency bands.  Figure \ref{fig:lfp_ratio_bands} plots the estimated  FS-ratio statistic $\widehat{R}_{i,a,b}$,  $i=1,2,\hdots,N=600$, for the known frequency bands: theta (4-8 Hertz), alpha (8-12 Hertz), beta (12-30 Hertz) 
and gamma (30-50 Hertz). At each epoch $i$, we obtained a $95\%$ confidence interval for the FS-ratio statistic using the  block bootstrap technique of \cite{politis94}. To select the block length, we follow the procedure in \citet{politis04,patton09}. Note that this procedure  is  for the univariate case and hence we apply it to each component of  the multivariate process $Y_{i,t}$ and obtain the block length as the average over all  components. The confidence intervals are the blue shaded region in Figures \ref{fig:lfp_ratio_bands}, \ref{fig:lfp_ratio_bands_specified}. 

The FS-ratio statistic is seen to have differences in the pre- and post-stroke epochs in the Theta, Alpha and Beta bands but not in the Gamma band. It can also be seen that the biggest difference in FS-ratio between pre- and post-stroke is in the Beta band wherein there is a decrease in the amount of spectral information after the stroke. Figure \ref{fig:lfp_ratio_bands_specified} also presents the FS-ratio statistic on other specified frequency bands wherein one notices differences between the pre- and post-stroke epochs.  

Tables 1 and 2 contain numerical summaries of the FS-ratio statistic for the pre- and post-stroke epochs at various frequency bands. We notice that the Beta band is where there is maximum difference observed between the pre- and post-stroke epochs. The Gamma band is consistent throughout the experiment's 600 epochs. Within the pre-stroke epochs (and also within the post-stroke epochs), the most variation in FS-ratio is observed in the Beta band.  

In \citet{fontaine_2019}, a univariate LFP microelectrode-wise change point  analysis was performed on the 
same dataset. In their work, for various frequency bands, changes in the non-linear spectral dependence of the LFP signal is modeled using parametric copulas. They detected change-points 
for a fixed microelectrode and fixed frequency band. One can notice the detection of numerous change points 
in the Delta, Theta, Alpha, Beta and Gamma bands for individual microelectrodes 1, 9 and 17. The detected 
change points include several epochs with very few of them being close to the time of the occlusion (or 
induced stroke) which was epoch $i=300$. 

In contrast, the advantages of our method are as follows: (i). The method treats the observed LFP signal as a multivariate nonstationary time series. Using \eqref{e:ssa_assumption_on_lfp}, we model this observed multivariate signal as a mixture of stationary and nonstationary components. Figure \ref{fig:lfp_dimension_plot} presents the dimension of stationary subspace (dimension of $Y_{i,t}$) across the 600 epochs and this is seen to be a useful feature in understanding changes in the cortical signal after the occurrence of the induced stroke (epoch 300). In other words, an increase in the dimension $d_i$ after the stroke points to a more stationary behavior of the LFP signal after the stroke. (ii).  The FS-ratio statistic, having the ability to compare two multivariate processes with unequal dimensions, is applied on the estimated processes $Y_{i,t}$ for each of the 600 epochs and frequency band specific numerical summaries are presented. The Beta frequency band is seen to be display the 
greatest changes the most within the pre stroke and post stroke epochs and also between the 
pre stroke and post stroke epochs. Also, from Figures \ref{fig:lfp_ratio_bands}, 
\ref{fig:lfp_ratio_bands_specified}, it is very easy to spot a change point at epoch 300 when the 
stroke was induced.

\begin{figure}[h]
\centering
\includegraphics[scale=0.45]{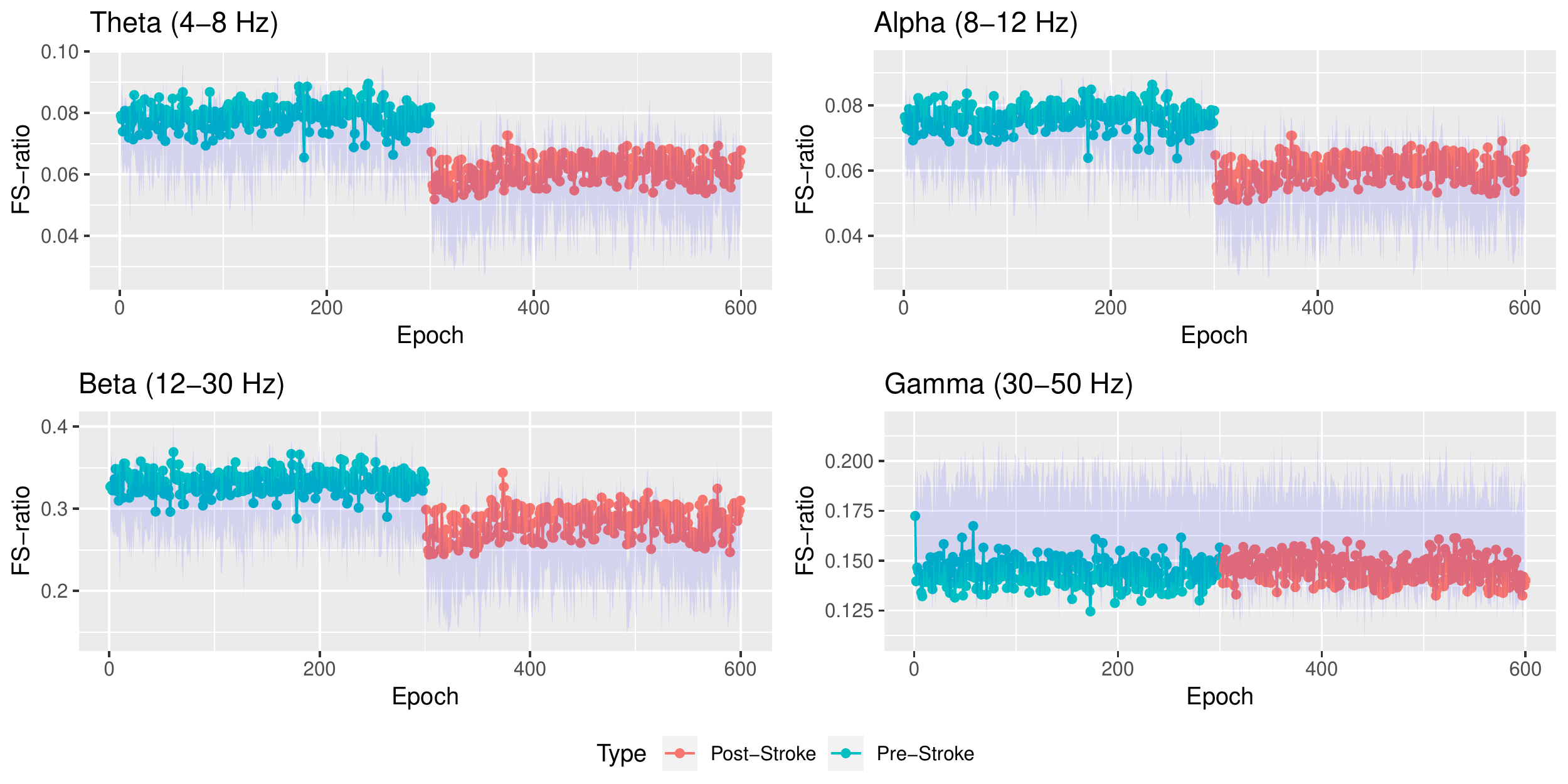}
\caption{ Plot of the FS-ratio statistic $\widehat{R}_{i,a,b}$ for $i=1,2,\hdots,N=600$ for various frequency bands. The blue shaded region corresponds to a 95\% confidence interval. } 
\label{fig:lfp_ratio_bands}
\end{figure}

\begin{figure}[h]
\centering
\includegraphics[scale=0.45]{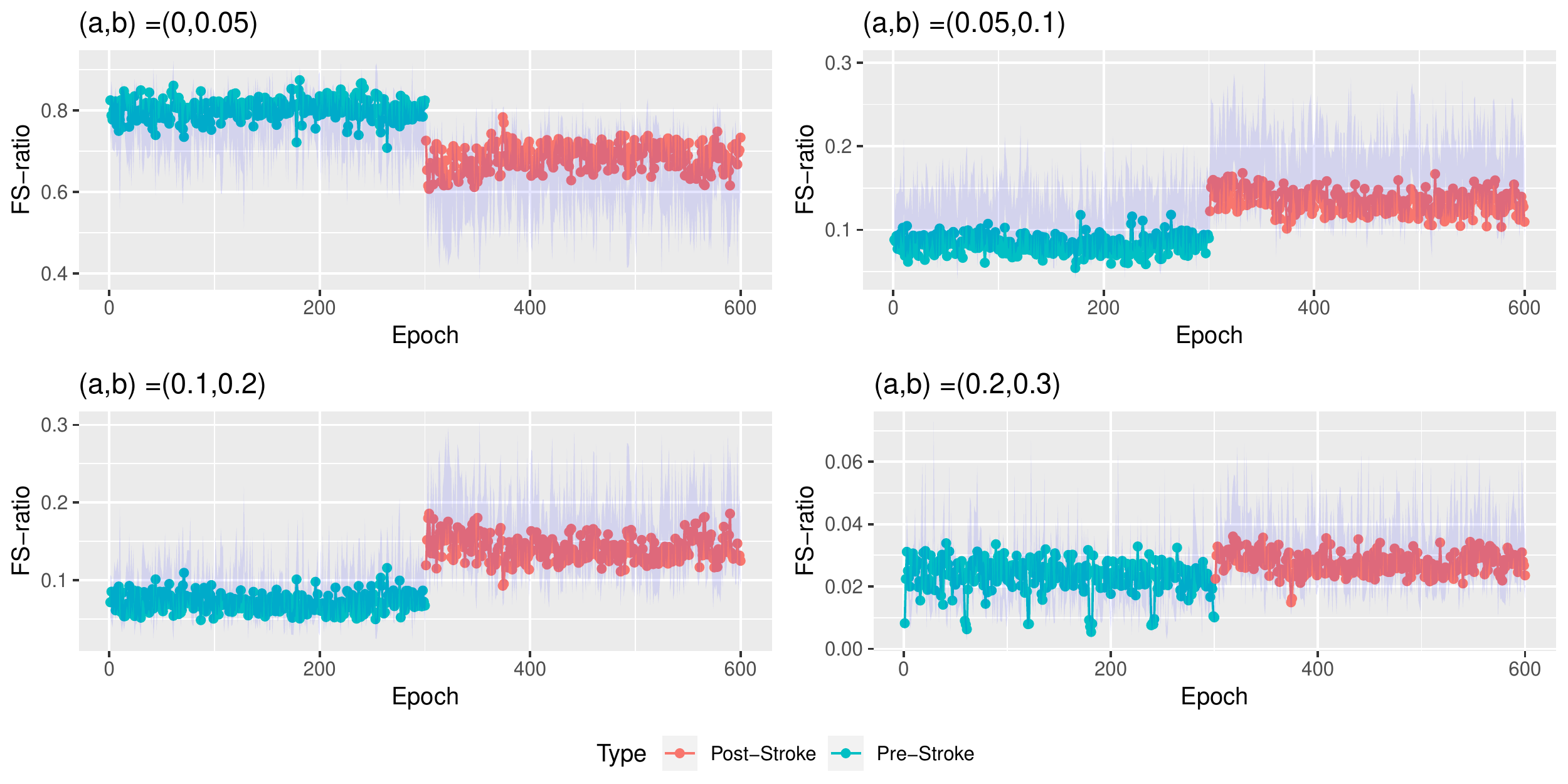}
\caption{ Plot of the FS-ratio statistic $\widehat{R}_{i,a,b}$ for $i=1,2,\hdots,N=600$ for specified frequency ranges $(a,b)$. Here $(a,b) \subset  (0,0.5)$ and $(0,0.5)$ corresponds to the interval $(0,\pi)$. The blue shaded region corresponds to a 95\% confidence interval.   } 
\label{fig:lfp_ratio_bands_specified}
\end{figure}

\begin{table}[H]
\begin{center}
\begin{tabular}{|c|c|c|c|c|c|}
\hline
Frequency Band & Mean & Median & SD & Lower & Upper \\
 &     &    & & CI & CI    \\        
\hline 
Theta (4-8 Hz)& 0.079 & 0.079 &  0.004 &  0.061  & 0.081 \\ 
\hline
Alpha (8-12 Hz) & 0.076 & 0.077 & 0.0035 & 0.059 & 0.078 \\
\hline
Beta (12-30 Hz) & 0.332 & 0.332 & 0.0129 & 0.267 & 0.341 \\
\hline
Gamma (30-50 Hz) & 0.144 & 0.144 & 0.006 & 0.141 & 0.191 \\
\hline
\end{tabular}
\end{center}
\vspace{-0.3cm}
\caption{Numerical summaries of FS-ratio statistic $\widehat{R}_{i,a,b}$ for \textbf{pre-stroke} epochs $i=1,2,\hdots,300$.}
\end{table}

\begin{table}[h]
\begin{center}
\begin{tabular}{|c|c|c|c|c|c|}
\hline
Frequency Band & Mean & Median & SD & Lower & Upper \\
 &     &    & & CI & CI    \\        
\hline 
Theta (4-8 Hz)& 0.062 & 0.062 & 0.004 & 0.0422 & 0.0669 \\ 
\hline
Alpha (8-12 Hz) & 0.060 & 0.061 & 0.004 & 0.041 & 0.064 \\
\hline
Beta (12-30 Hz) & 0.283 & 0.285 & 0.018 & 0.202 & 0.292 \\
\hline
Gamma (30-50 Hz) & 0.146 & 0.146 & 0.006 & 0.135 & 0.187 \\
\hline
\end{tabular}
\end{center}
\vspace{-0.3cm}
\caption{Numerical summaries of FS-ratio statistic $\widehat{R}_{i,a,b}$ for \textbf{post-stroke} epochs $i=301,302,\hdots,600$.}
\end{table}

\section{Concluding remarks} \label{s:conclusion}
In this work we proposed a new frequency-specific spectral ratio statistic FS-ratio that is 
demonstrated to be useful in comparing spectral information in two multivariate stationary processes of \underline{different dimensions}. The method is motivated by applications in neuroscience wherein brain signal is recorded across several epochs and the widely used tactic is to assume the observed signal be linearly generated by latent sources of interest in lower dimensions. Applying PCA/ICA/SSA to the observed signal in \underline{different epochs} in the experiment  results in \underline{different estimates of the dimensions} of latent sources. In these situations, the FS-ratio is seen to be useful because i). It captures the proportion of spectral information in various frequency bands by means of a $L_2$-norm on the spectral matrices and ii). It is blind to the dimension of the stationary process as it only looks at the proportion of spectral information at frequency bands. Under mild assumptions, the asymptotic properties of FS-ratio statistic are derived. In the application of our method to the LFP dataset, we witness the ability of our method in    
(i). identifying frequency bands where the pre- and post-stroke epochs are different, (ii). identifying   frequency bands that accounts for most variation within pre (and post) stroke epochs, (iii). identifying the frequency bands that are consistent across all the 600 epochs of the experiment and (iv). understanding the importance of contemporaneous dependence, both in the observed LFP and the stationary sources, across the 600 epochs and this indicated a perfect synchrony among microelectrodes after the stroke.

The theoretical results in Section \ref{s:theory} establish consistency of the FS-ratio statistic but further investigation is required to establish the large sample distribution. Such a result would help obtain confidence limits and also enable formal testing of the FS-ratio statistic. The result can also be useful in devising a change point method based on the FS-ratio statistic and that could be seen as a parallel to several change point methods in the literature applied on brain signal data; \citet{kirch_ombao}, \citet{schroeder_ombao}. The FS-ratio statistic carries out a multivariate analysis via $L_2$ norms on the spectral matrices. A related problem of interest would be to identify changes in  information in individual LFP microelectrodes. One approach could be modifying FS-ratio statistic by means of a weighted sum of the components of the spectral matrices.

\bibliographystyle{Chicago}
\bibliography{evol_spec}

\section*{Appendix: Proofs} \label{s:proofs}
Here we present the proofs of the theoretical results in Section \ref{s:theory}. 

\begin{proof}[\textbf{Proof of Theorem \ref{thm:r_ab_consistency}}]
Recall that for some $0<a<b<\pi$, 
\begin{align*}
\widehat{r}_{i,a,b}= \int_a^b || vec( \widehat{g}(\omega) ) ||_2^2 d \omega = \int_a^b || \frac{1}{T} \sum_{j=- \floor{T/2} }^{\floor{T/2}} \; K_h(\omega - \omega_j) vec(I_{ii}(\omega_j))||^2 \; d\omega \\
 =  \int_a^b \frac{1}{T^2} \sum_{ j_1,j_2 = - \floor{T/2} }^{\floor{T/2}} \; K_h(\omega - \omega_{j_1} ) K_h( \omega - \omega_{j_2} ) \sum_{r,s=1}^d \; I_{ii,rs}(\omega_{j_1}) \overline{I_{ii,rs}(\omega_{j_2})} \; d\omega.
\end{align*} 
We first consider the expected value of this quantity.
\begin{align*}
E(\widehat{r}_{i,a,b}) = \int_a^b \frac{1}{T^2} \sum_{ j_1,j_2 = - \floor{T/2} }^{\floor{T/2}} \; K_h(\omega - \omega_{j_1} ) K_h( \omega - \omega_{j_2} ) \sum_{r,s=1}^d \; E \Big( I_{ii,rs}(\omega_{j_1}) \overline{I_{ii,rs}(\omega_{j_2})} \Big) \; d\omega \\
 = \int_a^b \frac{1}{T^2} \sum_{ j_1,j_2 = - \floor{T/2} }^{\floor{T/2}} \; K_h(\omega - \omega_{j_1} ) K_h( \omega - \omega_{j_2} ) \sum_{r,s=1}^d \; g_{ii,rs}(\omega_{j_1}) \overline{g_{ii,rs}(\omega_{j_2})} \; d\omega \; + \; o(1).
\end{align*}
It can be seen that as $T \rightarrow \infty$, $h \rightarrow 0$ and $Th \rightarrow \infty $ the above quantity converges to 
\begin{align*}
\int_a^b \; \sum_{r,s=1}^d \Big( \int_{-\pi}^{\pi} K(v) dv \Big)^2  \; g_{ii,rs}(\omega) \; \overline{ g_{ii,rs}(\omega)} \; d\omega \; =  \; \int_a^b \; \sum_{r,s=1}^d  \; g_{ii,rs}(\omega) \; \overline{ g_{ii,rs}(\omega)} \; d\omega.
\end{align*}
Next, for the variance we have $V(\widehat{r}_{i,a,b}) = A_1 - A_2 $, where
\begin{align*}
A_1 = \int_a^b \int_a^b \frac{1}{T^4} \sum_{j_1,j_2,j_3,j_4} K_h(\omega - \omega_{j_1}) K_h(\omega - \omega_{j_2}) K_h(\lambda - \omega_{j_3}) K_h(\lambda - \omega_{j_4}) \; \times \\
 \sum_{r,s,t,u=1}^{d_i} E\Big( I_{ii,rs}(\omega_{j_1}) \overline{I_{ii,rs}(\omega_{j_2})} I_{ii,tu}(\omega_{j_3}) \overline{I_{ii,tu}(\omega_{j_4})} \Big) \; d\omega \; d\lambda \;\; \textrm{and}  \\
A_2 = \int_a^b \int_a^b \frac{1}{T^4} \sum_{j_1,j_2,j_3,j_4} K_h(\omega - \omega_{j_1}) K_h(\omega - \omega_{j_2}) K_h(\lambda - \omega_{j_3}) K_h(\lambda - \omega_{j_4}) \; \times \\
 \sum_{r,s,t,u=1}^{d_i} E\Big( I_{ii,rs}(\omega_{j_1}) \overline{I_{ii,rs}(\omega_{j_2})} \Big) E\Big( I_{ii,tu}(\omega_{j_3}) \overline{I_{ii,tu}(\omega_{j_4})} \Big) \; d\omega \; d\lambda.
\end{align*}
For the difference in the expectations between $A_1$ and $A_2$ we discuss the relevant cases and their convergence to 0. Firstly, it can be seen that for the following three cases the difference in the expectations is asymptotically 0:  a). $\omega_{j_1} \neq \omega_{j_2} \neq\omega_{j_3} \neq \omega_{j_4}$, b).  $\omega_{j_1} = \omega_{j_2} \neq\omega_{j_3} \neq \omega_{j_4}$, c).  $\omega_{j_1} = \omega_{j_2} \neq\omega_{j_3} = \omega_{j_4}$. Next, when $\omega_{j_1} = \omega_{j_3} \neq \omega_{j_2} = \omega_{j_4}$ we have,
\begin{gather*}
\int_a^b \int_a^b \frac{1}{T^4} \sum_{j_1,j_2} K_h(\omega - \omega_{j_1}) K_h(\omega - \omega_{j_2}) K_h(\lambda - \omega_{j_1}) K_h(\lambda - \omega_{j_2}) 
 \sum_{r,s,t,u=1}^{d_i} \Big[ E\Big( I_{ii,rs}(\omega_{j_1}) \overline{I_{ii,rs}(\omega_{j_2})} \times \\  I_{ii,tu}(\omega_{j_1}) \overline{I_{ii,tu}(\omega_{j_2})} \Big) 
  -  E\Big( I_{ii,rs}(\omega_{j_1}) \overline{I_{ii,rs}(\omega_{j_2})} \Big) 
   E\Big( I_{ii,tu} (\omega_{j_1}) \overline{I_{ii,tu}(\omega_{j_2})} \Big) \Big] d\omega \; d\lambda \\
   = \;
\int_a^b \int_a^b \frac{1}{T^4} \sum_{j_1,j_2} K_h(\omega - \omega_{j_1}) K_h(\omega - \omega_{j_2}) K_h(\lambda - \omega_{j_1}) K_h(\lambda - \omega_{j_2}) 
 \sum_{r,s,t,u=1}^{d_i} \Big[ E\Big( I_{ii,rs}(\omega_{j_1})I_{ii,tu}(\omega_{j_1}) \Big)  \times \\  
 E\Big( \overline{I_{ii,rs}(\omega_{j_2})}  \overline{I_{ii,tu}(\omega_{j_2})} \Big) 
  -  E\Big( I_{ii,rs}(\omega_{j_1}) \Big) E\Big(\overline{I_{ii,rs}(\omega_{j_2})} \Big) 
   E\Big( I_{ii,tu} (\omega_{j_1})\Big) E\Big( \overline{I_{ii,tu}(\omega_{j_2})} \Big) \Big] d\omega \; d\lambda + o(1) \\
   = \;
   \frac{1}{T^4} \sum_{j_1,j_2} \Big( \int_a^b  K_h(\omega - \omega_{j_1}) 
   K_h(\omega -  \omega_{j_2}) \; d\omega \Big)^2  \; 
 \sum_{r,s,t,u=1}^{d_i} \Big[ \Big( g_{ii,rt}(\omega_{j_1})g_{ii,su}(\omega_{j_1}) + g_{ii,rs}(\omega_{j_1})g_{ii,tu}(\omega_{j_1})\Big)\times \\
 \Big( \overline{g_{ii,rt}(\omega_{j_2})} \; \overline{g_{ii,su}(\omega_{j_2})} + \overline{g_{ii,rs}(\omega_{j_2})} \; \overline{g_{ii,tu}(\omega_{j_2})} \Big)  -  \Big(g_{ii,rs}(\omega_{j_1})\overline{g_{ii,rs}(\omega_{j_2})}
 g_{ii,tu}(\omega_{j_1})\overline{g_{ii,tu}(\omega_{j_2})} \Big)
  \Big]  \\
  \; + \;  o(1) \; =   \frac{1}{T^4h^2} \sum_{j_1,j_2} \Big( \int_{\frac{a-\omega_{j_1}}{h}}^{\frac{b-\omega_{j_1}}{h}}  K(u) 
   K(u +  \frac{\omega_{j_1}-\omega_{j_2}}{h}) du \Big)^2  \; 
 \sum_{r,s,t,u=1}^{d_i} \Big[ \cdots \Big] \; =  \; O(\frac{1}{T^2 h}). 
\end{gather*}
The case when $\omega_{j_1}=\omega_{j_2}=\omega_{j_3}\neq\omega_{j_4}$ would have the same rate of decay as above. Next, when $\omega_{j_1} = \omega_{j_3} \neq \omega_{j_2} \neq \omega_{j_4}$ we have,
\begin{gather*}
\int_a^b \int_a^b \frac{1}{T^4} \sum_{j_1,j_2,j_4} K_h(\omega - \omega_{j_1}) K_h(\omega - \omega_{j_2}) K_h(\lambda - \omega_{j_1}) K_h(\lambda - \omega_{j_4}) 
 \sum_{r,s,t,u=1}^{d_i} \Big[ E\Big( I_{ii,rs}(\omega_{j_1}) \overline{I_{ii,rs}(\omega_{j_2})} \times \\  I_{ii,tu}(\omega_{j_1}) \overline{I_{ii,tu}(\omega_{j_4})} \Big) 
  -  E\Big( I_{ii,rs}(\omega_{j_1}) \overline{I_{ii,rs}(\omega_{j_2})} \Big) 
   E\Big( I_{ii,tu} (\omega_{j_1}) \overline{I_{ii,tu}(\omega_{j_4})} \Big) \Big] d\omega \; d\lambda \\
   = \;\int_a^b \int_a^b \frac{1}{T^4} \sum_{j_1,j_2,j_4} K_h(\omega - \omega_{j_1}) K_h(\omega - \omega_{j_2}) K_h(\lambda - \omega_{j_1}) K_h(\lambda - \omega_{j_4}) 
 \sum_{r,s,t,u=1}^{d_i} \Big[ E\Big( I_{ii,rs}(\omega_{j_1}) I_{ii,tu}(\omega_{j_1}) \Big) \times  \\ 
 E\Big( \overline{I_{ii,rs}(\omega_{j_2})} \Big)  E\Big( \overline{I_{ii,tu}(\omega_{j_4})} \Big) 
  -  E\Big( I_{ii,rs}(\omega_{j_1}) \Big) E\Big(\overline{I_{ii,rs}(\omega_{j_2})} \Big) 
   E\Big( I_{ii,tu} (\omega_{j_1})\Big) E\Big( \overline{I_{ii,tu}(\omega_{j_4})} \Big) \Big] d\omega \; d\lambda + o(1) \\
   = \;
   \frac{1}{T^4} \sum_{j_1,j_2,j_4} \Big( \int_a^b  K_h(\omega - \omega_{j_1}) 
   K_h(\omega -  \omega_{j_2}) \; d\omega \Big)\Big( \int_a^b  K_h(\lambda - \omega_{j_1}) 
   K_h(\lambda -  \omega_{j_4}) \; d\lambda \Big)  
 \sum_{r,s,t,u=1}^{d_i} \Big[ \Big( g_{ii,rt}(\omega_{j_1}) \times \\
  g_{ii,su}(\omega_{j_1}) + 
 g_{ii,rs}(\omega_{j_1})g_{ii,tu}(\omega_{j_1}) \Big)\times 
 \Big( \overline{g_{ii,rs}(\omega_{j_2})} \; \overline{g_{ii,tu}(\omega_{j_4})} \Big) -  \Big(g_{ii,rs}(\omega_{j_1})\overline{g_{ii,rs}(\omega_{j_2})}
 g_{ii,tu}(\omega_{j_1})\overline{g_{ii,tu}(\omega_{j_4})} \Big) \Big]  \\  \; +  \; o(1)
  =  \frac{1}{T^4h^2} \sum_{j_1,j_2,j_4} \Big( \int_{\frac{a-\omega_{j_1}}{h}}^{\frac{b-\omega_{j_1}}{h}}  K(u) 
   K(u +  \frac{\omega_{j_1}-\omega_{j_2}}{h}) \; du \Big) \Big( \int_{\frac{a-\omega_{j_1}}{h}}^{\frac{b-\omega_{j_1}}{h}}  K(v) 
   K(v +  \frac{\omega_{j_1}-\omega_{j_4}}{h}) \; dv \Big) \times  \\
 \sum_{r,s,t,u=1}^{d_i} \; \Big[ \;  \cdots \;  \Big] \; + o(1) \; = \; O(\frac{1}{T}). 
\end{gather*}
Finally, we look at the case $\omega_{j_1} =\omega_{j_2} =\omega_{j_3}=\omega_{j_4}$. We have 
\begin{gather*}
\int_a^b \int_a^b \frac{1}{T^4} \sum_{j_1} K^2_h(\omega - \omega_{j_1})  K^2_h(\lambda - \omega_{j_1})
 \sum_{r,s,t,u=1}^{d_i} \Big[ E\Big( I_{ii,rs}(\omega_{j_1}) \overline{I_{ii,rs}(\omega_{j_1})} \times \\  I_{ii,tu}(\omega_{j_1}) \overline{I_{ii,tu}(\omega_{j_1})} \Big) 
  -  E\Big( I_{ii,rs}(\omega_{j_1}) \overline{I_{ii,rs}(\omega_{j_1})} \Big) 
   E\Big( I_{ii,tu} (\omega_{j_1}) \overline{I_{ii,tu}(\omega_{j_1})} \Big) \Big] d\omega \; d\lambda \\
 = \frac{1}{T^4}\sum_{j_1} \Big( \int_a^b K^2_h(\omega - \omega_{j_1}) \; d\omega \Big)^2 \sum_{r,s,t,u=1}^{d_i} \Big[ \cdots \Big] = \frac{1}{T^4 h^4}\sum_{j_1} \Big( \int_a^b K^2( \frac{\omega - \omega_{j_1}}{h} ) d\omega \Big)^2 \sum_{r,s,t,u=1}^{d_i} \Big[ \cdots \Big]  \\
  = \frac{1}{T^4h^2}\sum_{j_1} \Big( \int_{\frac{a-\omega_{j_1}}{h}}^{\frac{b-\omega_{j_1}}{h}}  K^2(u) du \Big)^2 \sum_{r,s,t,u=1}^{d_i} \Big[ \cdots \Big] = \; O(\frac{1}{T^3h^2})
\end{gather*}
\end{proof}


\begin{proof}[\textbf{Proof of Theorem \ref{thm:f_ij_testing}}]
Under the assumptions 1,2 stated earlier, asymptotic normality follows by the application of Theorem 3.5 of \citet{eichler08}. The mean and variance computations are similar to Theorem 2.1 in \citet{jentsch_pauly_15}.
\begin{gather*}
E\Big( \widehat{D}_{i,j} \Big) = 
  \frac{1}{T^2} \int_{a}^{b} \sum_{j_1,j_2} K_h(\omega - \omega_{j_1})K_h(\omega - \omega_{j_2}) \sum_{r,s=1}^{d_i} E \Big( (I_{ii,rs}(\omega_{j_1}) - I_{jj,rs}(\omega_{j_1})) \times ( \overline{ I_{ii,rs}(\omega_{j_2}) - I_{jj,rs}(\omega_{j_2}) } )   \Big) d\omega  \\
=  \frac{1}{T^2h^2} \int_{a}^{b} \sum_{j_1} K^2(\frac{\omega - \omega_{j_1}}{h}) \sum_{r,s=1}^{d_i}  \Big(     g_{ii,rr}(\omega_{j_1})\overline{g_{ii,ss}(\omega_{j_1}}) + g_{jj,rr}(\omega_{j_1})\overline{g_{jj,ss}(\omega_{j_1})} - g_{ij,rr}(\omega_{j_1})\overline{g_{ij,ss}(\omega_{j_1})} \\
 - g_{ji,rr}(\omega_{j_1}) \overline{g_{ji,ss}(\omega_{j_1})}  \Big) d\omega \; + \; o(1).
\end{gather*}
\noindent Hence $E(2 \pi T \sqrt{h} \; \widehat{D}_{i,j})$ is asymptotically equivalent to  
\begin{equation*}
\frac{1}{\sqrt{h}}\Big( \int_{-\pi}^{\pi}K^2(u) du \Big) \int_{-\pi}^{ \pi} 1_{\omega \in  (a,b)  } \Big(\; \sum_{p_1,p_2=1 }^{2} \big(\; -1 \; + \; 2 \delta_{p_1p_2} \; \big) |tr( G_{p_1p_2}(\omega) )|^2   \Big) \textrm{d}\omega     
\end{equation*}
For the variance, we have $V(\widehat{D}_{i,j}) \; = \;  B_1 - B_2$ where,
\begin{gather*}
B_1 = \int_a^b \int_a^b \frac{1}{T^4} \sum_{j_1,j_2,j_3,j_4} K_h(\omega - \omega_{j_1}) K_h(\omega - \omega_{j_2}) K_h(\lambda - \omega_{j_3}) K_h(\lambda - \omega_{j_4}) 
 \sum_{r,s,t,u=1}^{d_i} \\
  E\Big[ \Big( I_{ii,rs}(\omega_{j_1})
  -  
 I_{jj,rs}(\omega_{j_1}) \Big) \times 
 \Big( \overline{ I_{ii,rs}(\omega_{j_2}) - I_{jj,rs}(\omega_{j_2}) } \Big) \times 
  \Big( \overline{I_{ii,tu}(\omega_{j_3}) - I_{jj,tu}(\omega_{j_3})} \Big) \times \\
   \Big(  I_{ii,tu}(\omega_{j_4}) - I_{jj,tu}(\omega_{j_4})  \Big) \Big] \; d\omega \; d\lambda \;\; \textrm{and}  \\
B_2 = \int_a^b \int_a^b \frac{1}{T^4} \sum_{j_1,j_2,j_3,j_4} K_h(\omega - \omega_{j_1}) K_h(\omega - \omega_{j_2}) K_h(\lambda - \omega_{j_3}) K_h(\lambda - \omega_{j_4}) \; \times \\
 \sum_{r,s,t,u=1}^{d_i} E\Big[ \Big( I_{ii,rs}(\omega_{j_1})
  -  
 I_{jj,rs}(\omega_{j_1}) \Big) \times 
 \Big( \overline{ I_{ii,rs}(\omega_{j_2}) - I_{jj,rs}(\omega_{j_2}) } \Big) \Big] \times  \\
 E\Big[ \Big( \overline{ I_{ii,tu}(\omega_{j_3})
  -  
 I_{jj,tu}(\omega_{j_3})} \Big) \times 
 \Big(  I_{ii,tu}(\omega_{j_4}) - I_{jj,tu}(\omega_{j_4})  \Big) \Big]
   \; d\omega \; d\lambda.
\end{gather*}
\noindent The difference in the expectations is asymptotically non-zero when $\omega_{j_1}=\omega_{j_3}\neq \omega_{j_2}=\omega_{j_4}$, $\omega_{j_1}=-\omega_{j_3}\neq \omega_{j_2}=-\omega_{j_4}$, $\omega_{j_1}=\omega_{j_4}\neq \omega_{j_2}=\omega_{j_3}$, $\omega_{j_1}=-\omega_{j_4}\neq \omega_{j_2}=-\omega_{j_3}$. Considering only the first case with a factor of 4, the variance asymptotically yields
\begin{gather*}
V(\widehat{D}_{i,j}) =  \frac{4}{T^2h(2\pi)^2} 
  \int_{a - \pi}^{b+ \pi}\; \Big( \int_{- \pi}^{  \pi} K(u)K(u+v)du \Big)^2\; dv \int_{-\pi}^{\pi} 1_{\omega \in  (a,b)  } \Big( \; \sum_{p_1,p_2,p_3,p_4=1}^{2} (\; -1 \; + \; 2\delta_{p_1 p_2} \;) \\
  \;(\; -1 \; + \; 2\delta_{p_3 p_4}  \;) | tr(\; G_{p_1 p_3}(\omega) \overline{G_{p_2 p_4}(\omega)}^T \; )|^2    \Big) \textrm{d}\omega.
\end{gather*}

\end{proof}


\begin{proof}[\textbf{Proof of Theorem \ref{thm:R_consistency}}]
First, we look at the sufficient condition for joint consistency of $( \widehat{r}_{i,a,b} ,\widehat{r}_{i,\overline{\Pi}_{(a,b)}} )^{\top}$. Following the proof of Theorem \ref{thm:r_ab_consistency}, we have
 $cov(\widehat{r}_{i,a,b},\widehat{r}_{i,\overline{\Pi}_{(a,b)}}) = C_1 - C_2 $, where
\begin{align*}
C_1 = \int_{\overline{\Pi}_{(a,b)}} \int_a^b \frac{1}{T^4} \sum_{j_1,j_2,j_3,j_4} K_h(\omega - \omega_{j_1}) K_h(\omega - \omega_{j_2}) K_h(\lambda - \omega_{j_3}) K_h(\lambda - \omega_{j_4}) \; \times \\
 \sum_{r,s,t,u=1}^{d_i} E\Big( I_{ii,rs}(\omega_{j_1}) \overline{I_{ii,rs}(\omega_{j_2})} I_{ii,tu}(\omega_{j_3}) \overline{I_{ii,tu}(\omega_{j_4})} \Big) \; d\omega \; d\lambda \;\; \textrm{and}  \\
C_2 = \int_{\overline{\Pi}_{(a,b)}} \int_a^b \frac{1}{T^4} \sum_{j_1,j_2,j_3,j_4} K_h(\omega - \omega_{j_1}) K_h(\omega - \omega_{j_2}) K_h(\lambda - \omega_{j_3}) K_h(\lambda - \omega_{j_4}) \; \times \\
 \sum_{r,s,t,u=1}^{d_i} E\Big( I_{ii,rs}(\omega_{j_1}) \overline{I_{ii,rs}(\omega_{j_2})} \Big) E\Big( I_{ii,tu}(\omega_{j_3}) \overline{I_{ii,tu}(\omega_{j_4})} \Big) \; d\omega \; d\lambda.
\end{align*}
As in the proof of Theorem \ref{thm:r_ab_consistency}, it can be seen that, for the various cases, the covariance terms are of $O(\frac{1}{T^{\delta_1} h^{\delta_2}})$ where $\delta_1, \delta_2 \in  \{ 0,1,2,3 \}$ and $\delta_1 > \delta_2$. The result above along with Theorem \ref{thm:r_ab_consistency} implies 
\begin{equation*}
\Big( \widehat{r}_{i,a,b} ,\widehat{r}_{i,\overline{\Pi}_{(a,b)}}  \Big)^{\top} \; \xrightarrow[]{P} \; \Big( r_{i,a,b} , r_{i,\overline{\Pi}_{(a,b)}} \Big)^{\top}.
\end{equation*}
Finally, an application of the continuous mapping theorem yields the result. 
\end{proof}

\end{document}